\documentclass[a4paper]{article}

\usepackage[affil-it]{authblk}

\title{The pentagram map on Grassmannians}
\author{Ra\'{u}l Felipe \thanks{\texttt{raulf@cimat.mx}}}
\affil{CIMAT\\ Guanajuato, M\'exico}
\author{Gloria Mar\'{\i}~Beffa \thanks{\texttt{maribeff@math.wisc.edu}; Corresponding author}}
\affil{Mathematics Department\\
University of Wisconsin - Madison}

\date{\today}
\usepackage{amsmath,amsfonts,amsthm}
\usepackage{graphicx,pstricks}
\usepackage{color}

\newcommand{\R}{\mathbb{R}}
\newcommand{\CC}{\mathbb{C}}

\newcommand{\RP}{\mathbb{RP}}

\newcommand\SL{\mathrm{SL}}
\newcommand\Gr{\mathrm{Gr}}
\newcommand\GL{\mathrm{GL}}
\newcommand\ha{{b}}
\newcommand\hd{{q}}
\newcommand\hha{{\widehat{b}}}
\newcommand\Pm{\mathcal{P}}

\newcommand\rb{\mathbf{r}}
\newcommand\la{\lambda}
\newcommand\al{\alpha}
\newcommand\hG{\widehat{G}}
\newcommand\hg{\widehat{g}}
\newcommand\Tc{\mathcal{T}}
\newcommand\Tb{\overline{T}}
\newcommand\pb{\mathbf{p}}
\newcommand\ab{\mathbf{a}}

\newtheorem{theorem}{\textbf{Theorem}}

\newtheorem{example}{\textbf{Example}}
\newtheorem{definition}[theorem]{\textbf{Definition}}

\newtheorem{lemma}[theorem]{\textbf{Lemma}}

\newtheorem{remark}[theorem]{Remark}

\begin{document}

\maketitle



\begin{abstract} In this paper we define a generalization of the pentagram map to a map on twisted polygons in the Grassmannian space $\Gr(n,mn)$. We define invariants of Grassmannian twisted polygons under the natural action of $\SL(nm)$, invariants that define coordinates in the moduli space of twisted polygons. We then prove that when written in terms of the moduli space coordinates, the pentagram map is preserved by a  certain scaling. The scaling is then used to construct a Lax representation for the map that can be used for integration.
\end{abstract}

\section{Introduction}
In the last five years there has been a lot of activity around the study of the pentagram map, its generalizations and some related maps. The map was originally defined by Richard Schwartz over two decades ago (\cite{Sch1}) and after a dormant period it came back with the publication of \cite{OST}, where the authors proved that the map, when defined on {\it twisted} polygons, was completely integrable. The literature on the subject is quite sizable by now, as different authors proved that the original map on closed polygons was also completely integrable (\cite{OST2}, \cite{Sol}); worked on generalizations to polygons in higher dimensions and their integrability (\cite{KhS1}, \cite{KhS2}, \cite{MB1}, \cite{MB2}); and studied the integrability of other related maps (\cite{Sch2}, \cite{MB3}). The subject has also branched into geometry and  combinatorics, this bibliography refers only to some geometric generalizations of the map and is by no means exhaustive.

The success of the map is perhaps due to its simplicity. The original map is defined on closed convex polygons in $\RP^2$. The map takes a convex polygon in the projective plane to the one formed by the intersection of the lines that join every other vertex, as in the figure. The mathematical consequences of such a simple construction are astonishing (in particular, the pentagram map is a double discretization of the Boussinesq equation, a well-known completely integrable system modeling waves, see \cite{OST}). Integrability is studied not for the map itself, but for the map induced by it on the moduli space of planar projective polygons, that is on the space of equivalence classes of polygons up to a projective transformation. In \cite{OST} the authors defined it on {\it twisted} polygons, or polygons  with a monodromy after a period $N$, and proved that the map induced on the moduli space is completely integrable. (The map is equivariant under projective transformations, thus the existence of the moduli induced map is guaranteed.)

In this paper we look at the generalization of the map from the Grassmannian point of view. If we think of $\RP^2$ as the Grassmaniann $\Gr(1,3)$, that is, the space of homogeneous lines in $\R^3$, then the polygon would be a polygon in the Grassmannian under the usual action of $\SL(2+1)$, with each side representing a homogeneous plane as in the picture. The case of $\RP^{m-1}$ was studied in \cite{KhS1} where the authors proved that the generalized pentagram map was integrable for low dimensions, and conjectured that a scaling existed for the map that ensured the existence of a Lax pair and its integrability. The conjecture was proved in \cite{MB2}. We can also consider this case as $\Gr(1,m)$, $m\ge 3$. From this point of view, it is natural to investigate the generalized map defined on polygons in $\Gr(n, mn)$, where $m$ and $n$ are positive integers, $m\ge 3$. In this paper we define and study the generalization of the pentagram map to twisted Grassmannian polygons in $\Gr(n,mn)$, $m\ge 3$.

\vskip .2in
\centerline{\includegraphics[height=2.5in]{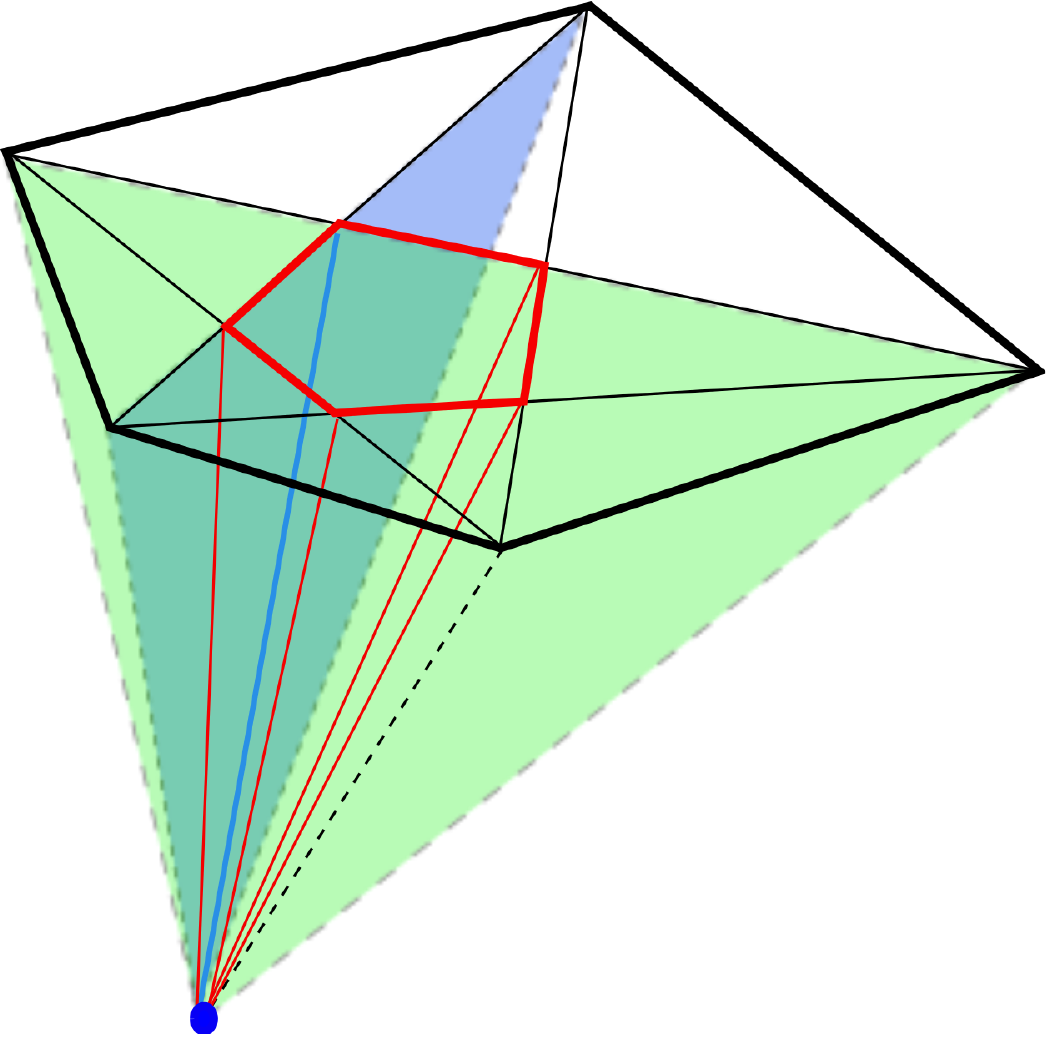}}
\vskip 2ex
\centerline{Figure 1: the pentagram map on pentagons in $\Gr(1,3)$}
\vskip .2in

The first step is to define the map on the moduli space of Grassmannian twisted polygons, that is, on the space of equivalence classes of Grassmannian twisted polygons, under the  classical action of $\SL((m-1)n+n)$ that generalizes the projective action of $PSL(m)$ on $\RP^{m-1}$. We do that by carefully studying the moduli space and finding generic coordinates that can be used to write the map in a convenient way (as in the case of the original pentagram map, the map can only be defined generically). The coordinates are found with the use of a discrete moving frame constructed through a normalization process similar to the one described in \cite{MMW}. The classification of invariants under this action  is, as far as we know, unknown, and it is completed in section 3.

In section 4 we study the case $m = 2s$.  In a parallel fashion to the study in \cite{MB2}, we proceed to write the pentagram map on the moduli space in the chosen coordinates, and we show that it can be written as the solution of a linear system of equations. We use that description and Cramer's rule to prove that the map is invariant under a certain scaling. As it was the case in \cite{MB2}, a critical part of the study is a fundamental lemma that decomposes the coefficient matrix of the system into terms that are homogeneous with respect to the scaling. This is lemma 4 for the even dimensional case, and lemma 9 for the odd dimensional one. The proofs of the rest of the results are supported by those two lemmas. Once the invariance under scaling is proved, the construction of a Lax representation is immediate when we introduce the scaling into a natural parameter-free Lax representation that exists for any map induced on the  moduli space by a map on polygons.

In section 5 we prove the case $m=2s+1$.

\section{Definitions and notations}

Let $\Gr(p,q)$ be the set of all $p$-dimensional subspaces of $V=\mathbb{R}^{q}$ or $V=\mathbb{C}^{q}$. Each  $l\in \Gr(p,q)$
can be represented by a matrix $X_{l}$ of size $q\times p$ such that the columns form a
 basis for $l$. We denote this relation by $l=<X_{l}>$. Clearly $l=<X_{l}>=<X_{l}d>$ for any $d\in \GL(p)$, and the
representation is not unique. Hence, $\Gr(p,q)$ can be viewed as the space of equivalence classes of $q\times p$ matrices, where two matrices are equivalent if their columns generate the same subspace. An element of this class, $X_{l}$ is called
{\it a lift of $l$}. The name reflects $\Gr(p,q)$ admitting the structure of a homogeneous space of dimension $p(q-p)$. Indeed, consider the Lie group $\SL(q+p)$, represented by block matrices of the form
\[
\begin{pmatrix} A_{q-p\times q-p} & B_{q-p\times p}\\ C_{p\times q-p} & E_{p\times p}\end{pmatrix}.
\]
 Let $H$ be the subspace defined by $B_{q-p\times p} = 0$. One can show that $\SL((q-p)+p)/H$ is isomorphic to $\Gr(p,q)$ and the natural action of $\SL((q-p)+p)$ on $\Gr(p,q)$ is  given by
 \[
 g\cdot <X> = < gX>.
 \]

Consider $\Gr(n,mn)$ for any positive integers $n,m$, and let  $\SL((m-1)n+n)\times \Gr(n,mn)\longrightarrow \Gr(n,mn)$ be the natural action of the group $\SL(mn)$ on $\Gr(n,mn)$.

A {\it twisted $N$-gon in $\Gr(n,mn)$} is a map $\phi:\mathbb{Z}\longrightarrow \Gr(n,mn)$ such that  $\phi(k+N)=M\cdot \phi(k)$ for all $k\in \mathbb{Z}$ and for some $M\in \SL(mn)$. The matrix $M$ is called the {\it monodromy} of the polygon and $N$ is the {\it period}.
We will also denote an $N$-gon by $\wp=(l_{k})$, where $l_{k}=\phi(k)$.

Let $X=(X_{k})$ be an arbitrary lift for an $N$-gon $\wp=(l_k)$ with monodromy $M$, and choose $X$ so it is also twisted, that is,  $X_{N+k} = M X_k$ for all $k$. For any discrete closed $N$-polygon $d=(d_{k})$ in $\GL(n)$ (i.e. satisfying $d_{k+N}=d_{k}$),
we have that $Xd=(X_{k}d_{k})$ is also a lift for the same polygon, with the same monodromy $M$. 

Let us denote by $\Pm_N$ the moduli space of twisted $N$-gons in $\Gr(n,mn)$, that is, the space of equivalence classes of twisted polygons under the natural action of $SL(mn)$. We will also denote by $\Pm l_N$ the moduli space of $N$-gons in $\R^{mn\times n}$ (or $\CC^{mn\times n}$, wherever the lifts live), under the linear action of $\SL(mn)$.


A $N$-gon $\wp = (l_k)$ is called \textbf{regular} if the matrix $\rho_k = (X_{k}\;X_{k+1}\ldots X_{k+m-2}\;X_{k+m-1})$  satisfies the
following condition
\begin{equation}\label{1}
\det \rho_k = |(X_{k}\;X_{k+1}\ldots X_{k+m-2}\;X_{k+m-1})|\neq 0,
\end{equation}
for any $k\in \mathbb{Z}$ and any lift $X$ (clearly, it suffices to check the condition for one particular lift). In other words, the columns of the matrix constitute a basis of $\mathbb{R}^{mn}$ (or
$\mathbb{C}^{mn}$) for all $k\in\mathbb{Z}$.


\section{The moduli space of twisted polygons in $\Gr(n, mn)$}

In this section we will prove that the moduli space of regular twisted polygons, $\Pm_N$, is a $N(m-1)n$-dimensional manifold and will define local coordinates.

 Assume $\wp = (l_k)$, $l_k\in \Gr(n, mn)$ is a regular twisted  $N$-gon and let $\{X_k\}$ be any twisted lift. By dimension counting, and given that $\wp$ is regular, for any $k=0,\dots,N-1$ we can find $n\times n$ matrices $a_k^i$, $i=0, \dots, m-1$ such that
 \begin{equation}\label{invariants}
 X_{k+m} = X_{m+k-1}a_k^{m-1} + \dots + X_{k+1} a_k^1 + X_k a_k^0.
 \end{equation}
Notice that if $\rho_k$ is as in (\ref{1}), then
\begin{equation}\label{K}
\rho_{k+1} = \rho_k \begin{pmatrix} O_n & O_n&\dots& O_n& a_k^0\\ I_n & O_n&\dots&O_n & a_k^1\\ \vdots&\ddots&\ddots&\vdots& \vdots\\ O_n & \dots& I_n&O_n&a_k^{m-2}\\ O_n & \dots& O_n&I_n&a_k^{m-1}\end{pmatrix} = \rho_k Q_k,
\end{equation}
where $Q_k$ is the matrix above. Using (\ref{1}), this implies that $\det a_k^0\ne 0$, for all $k$. Notice that $\rho_N = \rho_0 Q_0 Q_1, \dots, Q_{N-1}$. Thus, if $\rho_0 = I$, the monodromy is given by $M = Q_0 \dots Q_{N-1}$, and for other choices $M = \rho_0Q_0 \dots Q_{N-1}\rho_0^{-1}$. Thus, only the conjugation class of monodromy of the system defined by the matrices $Q_k$, $k=0,\dots,N-1$, is well-defined, not the monodromy itself.

 \begin{theorem}\label{invariantsth} Assume $m$ and $N$ are coprime. Then, for any regular twisted $N$-gon, $\wp$, there exists a lift $V = (V_k)$ such that
 \begin{equation}\label{norm1}
 \det(V_k, V_{k+1},\dots, V_{k+m-1}) = 1,
 \end{equation}
 for any $k=0,\dots, N-1$, and such that if $a_k^i$ are given as in (\ref{invariants}), then
 \begin{enumerate}
 \item\[
 a_k^0 = \mathrm{diagonal}(r_k^1, \dots, r_k^s),
 \]
 where each $r_k^s$ is an upper triangular Toeplitz matrix with $\det a_k^0 = 1$. 
 \item  We can choose $V$ such that all $a_k^{m-1}$'s entries, for any $k$, are generated by $N(n^2-n+1)$ independent functions.
 \end{enumerate} 
 The remaining $N(n-1)m$ entries of  $a_k^i$, $i\ne 0, m-1$, together with those above, define a coordinate system on $\Pm_N$.
 \vskip 1ex

 \begin{proof}{\rm

 Let $X = (X_k)$ be any twisted lift of the twisted polygon $\wp$. We will call $V_k = X_k d_k$ and show that we can find a closed polygon in $\GL(n)$, $\{d_k\}$, such that the conditions of the theorem are satisfied. If $V_k = X_k d_k$ we have
 \begin{equation}\label{norm2}
 (V_k, \dots, V_{k+m-1}) =  (X_k, \dots, X_{k+m-1})\mathrm{diag}(d_k, d_{k+1},\dots, d_{k+m-1}).
 \end{equation}
First of all we will show that condition (\ref{norm1}) determines the values of $\delta_k = \det d_k$, for any $k$. Indeed, from (\ref{norm2}) we have that
\[
 \prod_{i=k}^{k+m-1} \delta_i = Z_k,
\]
where $Z_k = \det(X_k, \dots, X_{k+m-1})^{-1}$ is determined by the choice of lift. These equations determine $\delta_k$ uniquely whenever $N$ and $m$ are coprime, as shown in \cite{MB2}. Let us call $\ha_k^i$ the invariants in (\ref{invariants}) associated to $X_k$ and $a_k^i$ those associated to $V_k$. Then, substituting in (\ref{invariants}), we have that
\begin{equation}\label{bnorm}
 a_k^i = d_{k+i}^{-1} \ha_k^i d_{k+m}.
\end{equation}
Let $\mathbf{p}=\{p_k\}$ be a closed polygon in $\GL(n)$ and define the $r$th $m$-product to be the product of every $m$ matrices starting at $p_r$ until we get to the end of the period, that is
 \[
 [p_r, \dots, p_{r+jm}]_m = p_rp_{r+m}p_{r+2m}\dots p_{r+jm},
 \]
with $r+(j+1)m \ge N$. If $N$ and $m$ are coprime, by repeatedly adding $m$ to the subindex we can reach all $N$ elements in $\{p_k\}$; that is, if $N$ and $m$ are coprime and $N = mq+s$, with $0<s<m$, then all $p_k$, $k=0,1,\dots, N-1$ appear in the product
\begin{equation}\label{pi}
\Pi_m(p_0) = [p_0, \dots, p_{N-s}]_m[p_{m-s}, \dots,]_m  \dots [,\dots ,p_{N+s-m}]_m[p_s,\dots, p_{N-m}]_m.
\end{equation}
(To see this one can picture a circle with $N$ marked points where we locate $p_j$. If we join with a segment every $m$ points, we are sure to join all points with segments before closing the polygon. If the polygon closes leaving some vertices untouched, it means that a multiple of $N$ can be divided into the union of disjoint orbits formed by joining every $m$ points. This would imply that $N$ and $m$ are not co-prime.)

Let us call
\begin{equation}\label{cr}
A_r = \Pi_m(a_r^0)= [a_r^0, a_{r+m}^0 \dots, a^0_{N+r-s}]_m  \dots [a^0_{s+r},\dots, a^0_{N-m+r}]_m.
\end{equation}
 We can see directly that $A_{r+m} = (a_r^0)^{-1}A_r a_r^0$, for all $r$. Once more, if $N$ and $m$ are coprimes, this property guarantees that {\it  all $A_r$ have the same Jordan form }, which we will call $J$.

Finally, notice that if $B_r = \Pi_m(b_r^0)$, then
\begin{equation}\label{norm3}
A_r = d_{r}^{-1} B_r d_{r}.
\end{equation}
Let us choose $d_{r}$ to be the matrix that conjugates $B_r$ to its Jordan normal form $J$, so that $A_r$ will all be in Jordan form. We can choose an order in the eigenvalues (for example, from smallest to largest) to ensure that the matrix is unique up to a factor that commutes with $J$. It is known that if a matrix commutes with a Jordan form matrix it must be  block diagonal
\[
\mathrm{diagonal}(r_1, \dots, r_s),
\]
where each $r_s$ is a Toepliz matrix, upper triangular, whenever the corresponding Jordan block is of the form
\[
\begin{pmatrix} \lambda & 1& 0&\dots & 0\\ 0&\lambda&1&\dots&0\\ \vdots & \ddots & \ddots& \dots&\vdots\\0&\dots&0&\lambda&1\\ 0& \dots&0&0&\lambda\end{pmatrix},
\]
or it is diagonal if the Jordan block is diagonal. Thus, $d_k$ are unique up to a block-diagonal matrix of this form.

Since
\[
B_{k+m} = (\ha_k^0)^{-1}B_k \ha_k^0,
\]
we have that
\[
A_{k+m} = J = d_{k+m}^{-1} B_{k+m} d_{k+m} = d_{k+m}^{-1} (\ha_k^0)^{-1}B_k \ha_k^0 d_{k+m}
\]
\[
= d_{k+m}^{-1} (\ha_k^0)^{-1}d_k d_k^{-1}B_k d_k d_k^{-1} \ha_k^0 d_{k+m} = (a_k^0)^{-1} A_k a_k^0= (a_k^0)^{-1} J a_k^0,
\]
for all $k$. Therefore, since $a_k^0$ commutes with the Jordan normal form, it must be a Toeplitz matrix of the form stated in the theorem, for all $k$.

 Finally, $d_k$ is unique up to a matrix commuting with $J$, lets call it $q_k$. We now turn our attention to the transformation of  $\ha_k^{m-1}$ under the change of lifting, namely
 \begin{equation}\label{ab}
 a_k^{m-1} = \hd_{k+m-1}^{-1} \hha_k^{m-1} \hd_{k+m},
 \end{equation}
where $\hha_k^{m-1} =d_{k+m-1}^{-1} \ha_k^{m-1} d_{k+m}$ ($d_k$ found above), and $\hd_r$ Toeplitz and commuting with $J$.

How to determine which entries generate the others depend very much on the particular point in the Grassmannian. In the generic case $\hd_r$ will all be diagonal; we will next describe the process generically. Using (\ref{ab}) we see that
\[
a^{m-1}_{k-m+1}a_{k-m+2}^{m-1}\dots a_{N+k-m}^{m-1} = q_k^{-1}\left(b_{k-m+1}^{m-1}b_{k-m+2}^{m-1}\dots b_{N+k-m}^{m-1}\right) q_k,
\]
for $k = 0,\dots, N-1$.

Before we describe the normalizations that will generate the syzygies, we recall that the determinants of $q_k$ are determined by (\ref{norm1}) for any $k=0,\dots, N-1$, whenever $N$ and $m$ are coprime. Let us call $\det d_k = \delta_k$. 

The last round of normalizations will be chosen  by equating those entries in place $(i, i+1)$, $i=1,2\dots,m-1$ with the entry $(2,1)$. 

If we denote by $q_k = \mathrm{diag}(q_k^1, \dots, q_k^n)$, and we denote the entries of \[b^k = b_{k-m+1}^{m-1}b_{k-m+2}^{m-1}\dots b_{N+k-m}^{m-1}\] by $b_{i,j}^k$ , then these normalizations result in equations of the form
\begin{equation}\label{uno}
\frac{q_k^{i+1}}{q_k^i} b^k_{i,i+1} = \frac{q_k^2}{q_k^1}b^k_{1,2},
\end{equation}
for $i=1,\dots,N-1$, and
\begin{equation}\label{dos}
\frac{q_k^{1}}{q_k^2} b^k_{2,1} = \frac{q_k^2}{q_k^1}b^k_{1,2}.
\end{equation}
Equation (\ref{dos}) solves for $q_k^2$ in terms of $q_k^1$ (if $b_{1,2}/b_{2,1}$ is not positive, we would need to choose different normalizations), and substituting it in (\ref{uno}) we get an expression for any $q_k^i$ in terms of $q_k^1$. Since $\det q_k = \det d_k^{-1}\delta_k$, where $d_k$ was determined in the normalization of $b_k^0$, $q_k^1$ is also determined.

 These last normalizations will produce as many syzygies in the entries in $b_k^{m-1}$ (linear or quadratic) as indicated in the statement of the theorem. The fact that the entries of $Q_k$, $k=0,\dots,N-1$, generate all other invariants of polygons in $\Gr(n,mn)$ is a consequence of the work in \cite{MMW}. 
} \end{proof}
 \end{theorem}

\begin{remark}{\rm It is very clear that these last normalizations could be chosen in many different ways (we could make entries constant, for example; or we could choose a different block, or relate entries from different blocks). Not all choices will work for us, and in order to be able to prove scaling invariance of the map, it is important that we choose the equations to be homogeneous in the entries of $b_k^r$. It is also simpler (although not necessary) if we choose entries
 from one block only to define the equations. The choice of $b_k^{m-1}$ versus $b_k^r$, $r\ne 0,m-1$ is just more convenient, but we could choose any other $r\ne 0$ instead.}
\end{remark}

\section{The Pentagram map on $\Gr(n,2sn)$}

\subsection{ Definition of the map }

Next, we define the Pentagram map for the Grassmannian $\Gr(n,2sn)$ for $s\geq 2$. The dimension of $\Gr(n,2sn)$ is clearly $(2s-1)n^2$.

Let $X=(X_{k})$ be a lift of a regular $N$-gon in $\Gr(n,2sn)$, and define the following subspaces
\begin{equation*}
\Pi_{k}=\langle X_{k},X_{k+2}\ldots,X_{k+2(s-1)},X_{k+2s}\rangle,
\end{equation*}
and
\begin{equation*}
\Omega_{k}=\langle X_{k+1},X_{k+3},\ldots,X_{k+2s-3},X_{k+2s-1}\rangle.
\end{equation*}

Note that $\dim \Pi_{k}=(s+1)n$ and $\dim \Omega_{k}=sn$. Therefore, generically, $\dim \Pi_{k}\cap \Omega_{k}=n$. 
\begin{definition}\label{pentadef1} Let $\wp=(l_{k})$ be a twisted $N$-gon in $G(n,2sn)$. Let $T(\wp)$ be the map taking the $N$-gon $\wp$ to the unique twisted $N$-gon whose vertices have a lift of the form
$T(X_{k})=\Pi_{k}\cap \Omega_{k}$. Notice that this is independent from the choice of the lift $X$. We call $T$ the Grassmannian Pentagram map.
\end{definition}

Notice that we are abusing notation by calling $T$ both the map on polygons and their lifts. We will go further and use the letter $T$ to denote the image of other data associated to $\wp$ in $T(\wp)$ (invariants, frames, etc). It is immediate to check that $T(\wp)$ is also twisted, with the same monodromy as $\wp$, using the fact that $\Pi_{N+k} = M \Pi_k$ and $\Omega_{N+k} = M \Omega_k$.

Next we will define this map in the moduli space of polygons, as defined by the invariants in our previous section. We will keep on using the letter $T$, defining $T: \mathcal{P}_{N}\longrightarrow \mathcal{P}_{N}$. We will specify the domain if needed. Let us assume that $V= (V_k)$ is the lift defined in theorem \ref{invariantsth} for a polygon $\wp$. Assume

\begin{equation}\label{22}
V_{k+2s}=V_{k}a_{k}^{0}+V_{k+1}a_{k}^{1}+\cdots+V_{k+2(s-1)}a_{k}^{2s-2}+V_{k+2s-1}a_{k}^{2s-1},
\end{equation}
as in (\ref{invariants})
for $2s$ matrices $a_{k}^{0}, a_{k}^{1},\ldots, a_{k}^{2s-1}$ of size $n\times n$ and with the properties described in theorem \ref{invariantsth}.

  Using the fact that $T(V_{k})\in \Pi_{k}$, we know that generically there exist matrices $c_i^j$ such
\begin{equation}\label{23}
T(V_{k})=V_{k}c_{k}^{0}+V_{k+2}c_{k}^{2}+\cdots\cdots+V_{k+2s}c_{k}^{2s}.
\end{equation}

If we now  use the relation (\ref{22}) we can
 replace $V_{k+2s}$ in (\ref{23}), and arrive to the following expression for $T(V_{k}):$
\begin{eqnarray*}
T(V_{k})&=&V_{k}c_{k}^{0}+V_{k+2}c_{k}^{2}+\cdots\cdots+V_{k+2s-2}c_{k}^{2s-2}\\
&+&\left(V_{k}a_{k}^{0}+V_{k+1}a_{k}^{1}+\cdots+V_{k+2(s-1)}a_{k}^{2s-2}+V_{k+2s-1}a_{k}^{2s-1}\right)c_{k}^{2s}\\
&=&V_k(c_k^0+a_k^0c_k^{2s})+V_{k+1}a_k^1c_k^{2s}+V_{k+2}(c_k^2+a_k^2c_k^{2s})+\dots + V_{k+2s-1}a_k^{2s-1}c_k^{2s}.
\end{eqnarray*}
Since we also assumed that $T(V_{k})\in \Omega_{k}$ then $c_{k}^{2r}+a_{k}^{2r}c_{k}^{2s}=0$, for $r=0,\dots,s-1$, and
\begin{equation}\label{26}
T(V_{k})=\left(V_{k+1}a_{k}^{1}+V_{k+3}a_{k}^{3}+\cdots +V_{k+2s-3}a_{k}^{2s-3}+V_{k+2s-1}a_{k}^{2s-1}\right)c_{k}^{2s}.
\end{equation}
Although the matrix $c_k^{2s}$ seems to be arbitrary, it is uniquely determined by the fact that the right hand side of (\ref{26}), not only for $k$, but also for  $k+1, \dots, k+2s-1$,  must be a lift for the image polygon $T(\wp)$, with the same properties as those found in theorem \ref{invariantsth}. Once $c_k^{2s}$ are chosen that way, we will be able to find the image of the matrices $a_i^j$ under the map $T$, as follows. Let us call $c_k^{2s} = \lambda_k$, so that we can write
\[
T(V_k) = \rho_k \rb_k\la_k,
\]
with $\rho_k = (V_k, V_{k+1}, V_{k+2}, \dots, V_{k+2s-1})$,
\begin{equation}\label{r}
\rb_k = \begin{pmatrix}O_n\\ a_k^1\\ O_n\\a_k^3\\\dots\\ O_n\\ a_k^{2s-1}\end{pmatrix},
\end{equation}
and where $\la_k$ are uniquely chosen so that $\{T(V_k)\}$ is the lift of the image polygon described in theorem \ref{invariantsth}.

\begin{example}{\rm In particular, (\ref{26}) implies that when $s=2$ the Pentagram map takes of following form
\begin{equation*}
T(V_{k})=[V_{k+1}a_k^1+X_{k+3}a_{k}^{3}]\la_k=(V_{k}\,V_{k+1}\,V_{k+2}\,V_{k+3})\begin{pmatrix} O_{n}\\ a_{k}^{1} \\ O_{n} \\ a_k^3 \end{pmatrix}\la_k = \rho_k\rb_k\la_k,
\end{equation*}
for all $k\in\mathbb{Z}$.}
\end{example}
Extending the map $T$ using, as usual, the pullback, we have that
\begin{eqnarray*}
T(\rho_k) &=& (T(V_k), \dots, T(V_{k+2s-1}))\\&=& \rho_k(\rb_k\la_k, R_k\rb_{k+1}\la_{k+1}, R_{k+1}\rb_{k+2}\la_{k+2}, \dots, R_{k+2s-2}\rb_{k+2s-1}\la_{k+2s-1}),
\end{eqnarray*}
where $R_{k+r} = Q_kQ_{k+1}\dots Q_{k+r}$ and $Q_k$ is given as in (\ref{K}). This expression can be written as
\begin{equation}\label{Lax2}
T(\rho_k) = \rho_k N_k \Lambda_k,
\end{equation}
where
\begin{equation}\label{N}
N_k = (\rb_k, R_k\rb_{k+1}, R_{k+1}\rb_{k+2}, \dots, R_{k+2s-2}\rb_{k+2s-1}),
\end{equation}
and
\begin{equation}\label{Lambda}
\Lambda_k = \begin{pmatrix} \la_k & O_n&\dots& O_n\\ O_n&\la_{k+1}&\dots&O_n\\ \ddots&\ddots&\ddots&\ddots\\ O_n&\dots& O_n&\la_{k+2s-1}\end{pmatrix}.
\end{equation}
One can recognize equations (\ref{K}) and (\ref{Lax2}), that is
\begin{equation}\label{Laxfree1}
\rho_{k+1} = \rho_k Q_k, \quad\quad T(\rho_k)=\rho_kN_k\Lambda_k,
\end{equation}
 as a parameter free Lax representation for the map $T$. The compatibility conditions  are given by
 \begin{equation}\label{compa}
 T(Q_k) = \Lambda_k^{-1}N_k^{-1} Q_k N_{k+1}\Lambda_{k+1}.
 \end{equation}
The last block-column of this equation defines the map $T$ on the moduli space of Grassmannian polygons, written in coordinates given by the invariants $a_i^j$.  The question we will resolve in the next subsection is how to introduce an spectral parameter in (\ref{Laxfree1}).

\subsection{A Lax representation for the pentagram map on $\Gr(n,2sn)$}

In this section we will prove that one can introduce a parameter $\mu$ in (\ref{Laxfree1}) in such a way that  (\ref{compa}) will be independent from $\mu$. This will define a true Lax representation that can be used for integration of the map. As it was done in \cite{MB2}, we will prove that the map $T$ is invariant under the scaling
\begin{equation}\label{scaling1}
 \quad\quad a_k^{2r+1} \to \mu a_k^{2r+1}, \quad\quad a_k^{2r} \to  a_k^{2r},
 \end{equation}
 for any $r= 0,1,\dots,s-1$ and any $k$ (this implies that all entries of these $n\times n$ matrix scale equally). This will involve several steps.

Let us denote the block columns of (\ref{N}) by $F_r$, so that $F_k = \rb_k$, and
\begin{equation}\label{F}
F_{k+\ell} = R_{k+\ell-1}\rb_{k+\ell} = Q_{k}Q_{k+1}\dots Q_{k+\ell-1}\rb_{k+\ell},
\end{equation}
$\ell=1,2,\dots$, with $\rb$ as in (\ref{r}). Our first lemma will allow us to decompose the block columns of $N_r$ into homogeneous terms according to (\ref{scaling1}). The lemma is almost identical to Lemma 3.1 in \cite{MB2}. Let us denote by $\Gamma$ the matrix
\begin{equation}\label{Gamma}
\Gamma = \begin{pmatrix}O_n&O_n&O_n&\dots&O_n\\ I_n&O_n&O_n&\dots&O_n\\ O_n&I_n&O_n&\dots&O_n\\ \vdots&\ddots&\ddots&\ddots&\vdots\\ O_n&\dots&O_n&I_n&O_n\end{pmatrix},
\end{equation}
where $I_n$ is the $n\times n$ identity matrix. Let us also denote by $\Tc$ the shift operator, namely $\Tc(V_k) = V_{k+1}$. This shift operator can trivially be extended to invariants using $\Tc(a_k^i) = a_{k+1}^i$ and to functions depending on the invariants using the pullback. We can also extend it to matrices whose entries are invariants by applying it to each entry, as it is customary.
\begin{lemma}\label{Mainlemma}
Let $F_{k+\ell}$ be given as in (\ref{F}). Then, there exist $n\times n$ matrices $\al_i^j$ such that
\begin{equation}\label{maineq}
F_{k+2\ell} = \sum_{r=1}^\ell F_{k+2r-1}\al_{2r-1}^{2\ell} + G_{k+2\ell}, \quad F_{k+2\ell+1} = \sum_{r=0}^\ell F_{k+2r}\al_{2r}^{2\ell+1}+\hG_{k+2\ell+1},
\end{equation}
for $\ell\ge 1$, where
\begin{equation}\label{Godd}
\hG_{k+2\ell+1} = \pb_k\left(\Tc G_{k+2\ell}\right)_m+\Gamma \Tc G_{k+2\ell}, \quad \al_{2r}^{2\ell+1} = \Tc\al_{2r-1}^{2\ell}, \quad \al_0^{2\ell+1} =\left(\Tc G_{k+2\ell}\right)_m,
\end{equation}
with
\begin{equation}\label{p}
\pb_k = \begin{pmatrix}a_k^0\\ O_n\\ a_k^2\\ O_n\\ \dots\\ a_k^{m-2}\\ O_n\end{pmatrix},
\end{equation}
 and
\begin{equation}\label{Geven}
G_{k+2\ell+2} = \Gamma \Tc\hG_{k+2\ell+1}, \quad \al_{2r+1}^{2\ell+2} = \Tc\al_{2r}^{2\ell+1}, \quad G_k = F_k = \rb_k.
\end{equation}
By $A_m$ we mean the last $n\times n$ block entry of a matrix $A$.
\end{lemma}

From now on we will simplify our notation by denoting $F_{k+\ell}$ simply by $F_\ell$. We will introduce the subindex $k$ only if its removal creates confusion.
\begin{proof}
First of all, notice that the last column of $Q$ in (\ref{K}) is given by $\pb+\rb$, as in (\ref{p}) and (\ref{r}). Notice also that, from the definition in (\ref{F}) we have
\[
F_\ell = Q\Tc F_{\ell-1}.
\]
We proceed by induction. First of all, since $F = \rb$,
\[
F_1 = Q\Tc F = Q\rb_1 = (\pb+\rb)a_1^{m-1} + \Gamma \rb_1 = F a_1^{m-1}+\pb a_1^{m-1} + \Gamma\rb_1.
\]
We simply need to call $\hG_1 = \pb a_1^{m-1}+\Gamma r_1$, and $\al_0^1 = a_1^{m-1} = \Tc\left(F\right)_m$. Let's do the first even case also:
\[
F_2 = Q\Tc F_1 = Q(\Tc \hG_1+\Tc F a_2^{m-1}) = F_1 a_2^{m-1} + \Gamma \Tc\hG_1,
\]
and we call $G_2 = \Gamma\Tc \hG_1$. Notice that $Q\Tc\hG_1 = \Gamma\Tc \hG_1$ since the last block of $\hG_1$ vanishes, as indicated by the hat.

Now, assume
\[
F_{2\ell} = \sum_{r=1}^\ell F_{2r-1}\al_{2r-1}^{2\ell} + G_{2\ell}.
\]
Then
\[
F_{2\ell+1} = Q \Tc F_{2\ell}= \sum_{r=1}^\ell Q\Tc F_{2r-1} \Tc \al_{2r-1}^{2\ell} + Q\Tc G_{2\ell}.
\]
Since $Q\Tc G_{2\ell} =  \pb \left(\Tc G_{2\ell}\right)_m+\rb \left(\Tc G_{2\ell}\right)_m + \Gamma \Tc G_{2\ell}$ and $\rb = F$, if
\[
\hG_{2\ell+1} = \pb \left(\Tc G_{2\ell}\right)_m+ \Gamma \Tc G_{2\ell}, \quad \al_{2r}^{2\ell+1} = \Tc \al_{2r-1}^{2\ell},~ r=1,\dots,\ell, \quad  \al_0^{2\ell+1} = \left(\Tc G_{2\ell}\right)_m,
\]
then we have
\[
F_{2\ell+1} = \sum_{r=0}^\ell F_{2r} \al_{2r}^{2\ell+1} + \hG_{2\ell+1}.
\]
Looking into the even case, we have that
\[
F_{2\ell+2} = Q\Tc F_{2\ell+1} = \sum_{r=0}^\ell Q\Tc F_{2r}\Tc \al_{2r}^{2\ell+1} + Q \Tc \hG_{2\ell+1},
\]
and since $F_{2r+1} = Q\Tc F_{2r}$ and $Q\Tc \hG_{2\ell+1} = \Gamma \Tc \hG_{2\ell+1}$, if we call
\[
G_{2\ell+2} = \Gamma \Tc \hG_{2\ell+1}, \quad \al_{2r+1}^{2\ell+2} = \Tc \al_{2r}^{2\ell+1},
\]
we prove the lemma.
\end{proof}

Once we have this lemma we can identify homogeneous terms in the expansion of block columns. Indeed, notice that $\rb$ and $\pb$ are homogeneous of degree $1$ and $0$, respectively, with  respect to the scaling. Since the shift clearly preserves the degree, from the statement of the lemma we have that both $\hG_{k+2r+1}$ and $G_{2r}$ are homogeneous of degree $1$, for any $r$. Likewise, $\al_0^r$ are also homogeneous of degree $1$, for any $r$, from its definition, and since all others are obtained by shifting these, they are also.

Therefore, if we denote $G = F$, iteratively applying the lemma we have that
$F_r$ are in all cases a combination of $G_{2r}$ and $\hG_{2r+1}$ for the different values of $r$, with different types of factors of the form $\al_i^j$, each of degree $1$. One can also clearly see that if the columns of $F_r$ generate $R^{nm}$ for $r=0, 1,\dots, m-1$, them the columns of $G_{2r}$ and $\hG_{2r+1}$, $r=0, \dots,s-1$ will also generate the same space since the change of basis matrix will be upper triangular with ones down the diagonal. This new basis will be crucial in the calculations that follow.

Finally, a comment as to the reason for our notation. Notice that both $G_{2\ell}$ and $\hG_{2\ell+1}$ have alternative zero and non-zero clocks, with $G_{2\ell}$ starting with a zero block and $\hG_{2\ell+1}$ starting with a nonzero block. We are keeping that marked not only by the subindex but also by a hat, since as calculations become more involved it helps to have them be visibly different. It shows that the entire space can be written as a direct sum of two orthogonal subspaces, one generated by the block columns with hats and one generated by those without hats.

Next, assume that we drop the $\Lambda_k$ factor and define
\[
\Tb(V_k) = \rho_k \rb_k.
\]
Define further $c_k^i = \Tb(a_k^i)$,  as given by the following compatibility formula, which is (\ref{compa}) after removing $\Lambda_k$
\begin{equation}\label{companol}
\Tb(Q_k) = N_k^{-1} Q_k N_{k+1}.
\end{equation}
Notice that $c_k^i$ will need to be normalized by $\Lambda_k$ before we can declare it to be $T(a_k^i)$. Let us call $\ab_k$ the last block column in $Q_k$ (the $ith$ block will be $a_k^{i-1}$). Then, choosing the last block-column in both sides of the equation
\[
\Tb(\ab_k) = N_k^{-1} Q_k \Tc F_{k+2s-1} = N_k^{-1} F_{k+2s},
\]
which can be written as
\begin{equation}\label{lineareq}
 N_k \Tb(\ab_k) = F_{k+2s}.
\end{equation}
Thus {\it $\Tb(\ab_k)$ can be interpreted as the solution of the linear equation (\ref{lineareq})}. This will be crucial in what follows.

\begin{theorem}
The matrices $c_k^{i}$ are homogeneous with respect to the scaling (\ref{scaling1}), and
\[
d(c_k^{2\ell}) = 0, \quad\quad d(c_k^{2\ell+1}) = 1,
\]
for any $\ell=0,1,\dots,s-1$.
\end{theorem}
\begin{proof}

As in the previous proof, we will drop the subindex $k$ and introduce it only if needed.

First of all, let us analyze the homogeneity with respect to (\ref{scaling1}) of the determinant
\[
D = \det N = \det(F, F_1, \dots, F_{2s-1}).
\]
From (\ref{maineq}) we can rewrite it as
\[
D = \det(F, \hG_1, G_2, \dots, G_{2s-2}, \hG_{2s-1}),
\]
and since $d(G_{2\ell})  =d(\hG_{2\ell-1}) = 1$ for all $\ell$, we have that $D$ is homogeneous and $d(D) = 2sn$.

Next, denote by $f_{r}^j$ the $j$th column of $F_r$, and let $F_{r,i}^j$ be the block column whose individual columns are equal to those of $F_r$, except for the $i$th column which is equal to the $j$th column of $F_{2s}$, for any $r=0,1,\dots,2s-1$.

Define next
\[
D_{2\ell, i}^j = \det(F, F_1, \dots, F_{2\ell-1}, F_{2\ell,i}^j, F_{2\ell+1}, \dots, F_{2s-1}).
\]
We first notice that using (\ref{maineq}), we can substitute the $i$th column of $F_{2\ell,i}^j$ by the $j$th column of $G_{2s}$, since $F_{2s}$ and $G_{2s}$ differ in a linear combination of columns of $F_{2r+1}$, $r< s$. Let us call the new matrix $F_{2\ell,i}^{j,g}$.  We then simplify the part of the determinant to the right of $F_{2\ell,i}^j$ using (\ref{maineq}), to become
\[
D_{2\ell, i}^j = \det(F, F_1, \dots, F_{2\ell-1}, F_{2\ell,i}^{j,g}, \hG_{2\ell+1}+ f_{2\ell}^i\left(\al_{2\ell}^{2\ell+1}\right)_i, G_{2\ell}, \dots, G_{2s-2},\hG_{2s-1} + f_{2\ell}^i\left(\al_{2\ell}^{2s-1}\right)_i),
\]
where $\left(\al_{2\ell}^{r}\right)_i$ denotes the $i$th row of $\al_{2\ell}^{r}$.

We now proceed to simplify the columns $f_{2\ell}^i$ which can be substituted by $g_{2\ell}^i$ (the $i$th column of $G_{2\ell}$) since their difference is generated by odd vectors with subindiced less that $2\ell$. We can then simplify the rest of the determinant, using (\ref{maineq}) once more. We get
\[
D_{2\ell, i}^j = \det(F, \hG_1, \dots, \hG_{2\ell-1}, G_{2\ell,i}^j, \hG_{2\ell+1}+ g_{2\ell}^i\left(\al_{2\ell}^{2\ell+1}\right)_i, G_{2\ell}, \dots, G_{2s-2},\hG_{2s-1} + g_{2\ell}^i\left(\al_{2\ell}^{2s-1}\right)_i),
\]
where $G_{2\ell,i}^j$ indicates the matrix equal to $G_{2\ell}$, except for the $i$th column which is equal to $g_{2s}^j$. Our last step is to notice that we have enough $G_{2r}$ block-columns ($r=0,\dots,s-1$, $r\ne s$) that together with $G_{2\ell,i}^j$ generically generate the entire subspace generated by $G_{2r}$, $r=0,\dots,s-2$. But the vector $g_{2\ell}^i$ belongs to this subspace, and hence it will be a combination of the columns of those blocks. Thus
\[
D_{2\ell, i}^j = \det(F, \hG_1, \dots, \hG_{2\ell-1}, G_{2\ell,i}^j, \hG_{2\ell+1}, G_{2\ell}, \dots, G_{2s-2},\hG_{2s-1}),
\]
which clearly shows that $D_{2\ell, i}^j$ is homogeneous, and since $d(G_{2\ell,i}^j) = 1$, $d(D_{2\ell, i}^j) = 2sn$ also.

Finally, $\Tb(\ab)$ is the solution of system of linear equations (\ref{lineareq}), and so, by Cramer's rule, the $(j,i)$ entry of $c^{2\ell} = \Tb(a^{2\ell})$ is of the form
\[
\frac{D_{2\ell,i}^j}D.
\]
Therefore, $c^{2\ell}$ is homogeneous and $d(c^{2\ell}) = 0$, $\ell=0,\dots,s-1$.

We now study $c^{2\ell+1}$. Consider  the determinant
\[
D_{2\ell+1, i}^j = \det(F, F_1, \dots, F_{2\ell}, F_{2\ell+1,i}^j, F_{2\ell+2}, \dots, F_{2s-1}),
\]
where $F_{2\ell+1,i}^j$ is defined as $F_{2\ell+1}$ substituting the $i$th column with the $j$th column of $F_{2s}$. As before, using (\ref{maineq}), we start by noticing that we can substitute the $i$th column of $F_{2\ell+1,i}^j$ by the $j$th column of $G_{2s}$, call it $g_{2s}^j$, plus $f_{2\ell+1}^i\left(\al_{2\ell+1}^{2s}\right)_{i,j}$ that comes from the expansion of $F_{2s}$ in terms of odd terms, and the fact that the $i$th column of $F_{2\ell+1}$ is missing.  The expression $\left(\al_{2\ell+1}^{2s}\right)_{i,j}$ is the $(i,j)$ entry of $\al_{2\ell+1}^{2s}$. We call the resulting matrix $F_{2\ell+1,i}^{j,g}$. If we simplify the right hand side of the determinant it becomes
\begin{eqnarray*}
&&D_{2\ell+1, i}^j \\
&=& \det(F, F_1, \dots, F_{2\ell}, F_{2\ell+1,i}^{j,g}, G_{2\ell+2}+ {f}_{2\ell+1}^i\left(\al_{2\ell+1}^{2\ell+2}\right)_i, \hG_{2\ell+3}, \dots, G_{2s-2}+{f}_{2\ell+1}^i\left(\al_{2\ell+1}^{2s-2}\right)_i,\hG_{2s-1})\\
&=& \left(\al_{2\ell+1}^{2s}\right)_{i,j}\det(F, F_1, \dots, F_{2\ell}, F_{2\ell+1},G_{2\ell+2},\dots, \hG_{2s-1})\\
&+& \det(F, F_1, \dots, F_{2\ell}, \hG_{2\ell+1,i}^{j}, G_{2\ell+2}+ {f}_{2\ell+1}^i\left(\al_{2\ell+1}^{2\ell+2}\right)_i, \hG_{2\ell+3}, \dots, G_{2s-2}+{f}_{2\ell+1}^i\left(\al_{2\ell+1}^{2s-2}\right)_i,\hG_{2s-1}),
\end{eqnarray*}
where $\hG_{2\ell+1,i}^j$ is equal to $\hG_{2\ell+1}$ except for the $i$th column which is equal to $g_{2s}^j$.

As before, $f_{2\ell+1}^i$ and $\widehat{g}_{2\ell+1}^i$ differ in a sum of columns of $F_{2r}$, $r\le \ell$. Thus, we can substitute $f_{2\ell+1}^i$ by $\widehat{g}_{2\ell+1}^i$ in the determinant. After that, we proceed to simplify the rest of the determinant obtaining
\begin{eqnarray*}
&&D_{2\ell+1, i}^j = \left(\al_{2\ell+1}^{2s}\right)_{i,j} D\\
&+&  \det(F, \hG_1, \dots, G_{2\ell}, \hG_{2\ell+1,i}^{j}, G_{2\ell+2}+ \widehat{g}_{2\ell+1}^i\left(\al_{2\ell+1}^{2\ell+2}\right)_i, \hG_{2\ell+3}, \dots, G_{2s-2}+\widehat{g}_{2\ell+1}^i\left(\al_{2\ell+1}^{2s-2}\right)_i,\hG_{2s-1}),
\end{eqnarray*}
Unlike the previous case, this time one of the columns of $ \hG_{2\ell+1,i}^{j}$ is even, and  hence, the odd columns (other than $\widehat{g}_{2\ell+1}^i$) do not generate the odd orthogonal subspace since they are one dimension short. Thus, including $\widehat{g}_{2\ell+1}^i$, we have an equal number of odd and even columns and we need to expand.

The term that includes no $\widehat{g}_{2\ell+1}^i$ in the expansion is given by
\[
\det(F, \hG_1, G_2, \dots, G_{2\ell}, \hG_{2\ell+1,i}^j, G_{2\ell+2}, \dots G_{2s-2}, \hG_{2s-1})  =0,
\]
since, as we said before, there are more even columns that odd columns.
The remaining terms in the expansion are
\[
\sum_{r=\ell+1}^{s-1}\sum_{p=1}^n\left(\al_{2\ell+1}^{2r}\right)_{i,p}\det(F, \hG_1, \dots, \hG_{2\ell+1,i}^j, \dots, \hG_{2r-1}, G_{2r}^i+\widehat{g}_{2\ell+1}^i e_p^T, \hG_{2r-1}, \dots, G_{2s-2}, \hG_{2s-1}),
\]
where $\left(\al_{2\ell+1}^{2r}\right)_{i,p}$ is the $(i,p)$ entry of $\al_{2\ell+1}^{2r}$, $G_{2r}^i$ has zero $i$th column and where $e_p$ is the standard canonical basis of $\R^n$ with a $1$ in the $p$th entry and zero elsewhere.   Each one of these determinants has an equal number of odd and even columns. Each column is homogeneous of degree $1$, and so each determinant is homogeneous of degree $2sn$. But, like $D$ (also of degree $2sn$), they are multiplied by $\left(\al_{2\ell+1}^{2r}\right)_{i,p}$, homogeneous of degree $1$. Hence, $D_{2\ell+1, i}^j $ is homogeneous and $d(D_{2\ell+1, i}^j) = 2sn+1$.

Finally,  since, according to (\ref{lineareq}), the $(j,i)$ entry of $\Tb(a^{2\ell+1}) = c^{2\ell+1}$ is equal to
\[
\frac{D_{2\ell+1, i}^j }D,\]
we conclude that $c^{2\ell+1}$ is homogeneous and $d(c^{2\ell+1}) = 1$. This concludes the proof of the theorem.
\end{proof}

Our final step is to introduce the normalization matrices $\la_k$ and to study how they might affect the scaling degree of $T(a_k^i)$. Recall that $\la_k$ has two factors: $d_k$, used to normalize $c_k^0$ and to transform them into their Jordan form (as in (\ref{bnorm})); and $q_k$, in the generic case, a diagonal matrix used to define syzygies among the entries of $c_k^{m-1}$, as in (\ref{ab}) ($c_k^{m-1}$ plays the role of $b_k^{m-1}$ in (\ref{ab})).

\begin{lemma} The matrices $\la_k$ are homogeneous with respect to (\ref{scaling1}) and
\[
d(\la_k) = -1,
\]
for all $k$.
\end{lemma}
\begin{proof}
Since $\la_k = d_k q_k$, we will look at each factor separately.

The first factor $d_k$ is determined by the normalization of $B_k$ as in (\ref{norm3}), where, in our case, $B_k = \Pi_m(c_k^0)$ as in (\ref{pi}). But, given that $c_k^0$ are invariant under the scaling, $B_k$ will also be, and hence so will $d_k$.

The second factor, $q_k$,  is found by using a number of equations of the form (\ref{uno})-(\ref{dos}), which finds each entry of $q_k$ as fuctions of the first entry $q_k^1$. Since $c_k^{m-1}$ (which plays the role of $b_r^{m-1}$) is homogeneous with respect to the scaling, equations (\ref{uno})-(\ref{dos}) imply that $q_k^i$ are homogeneous also, with equal degree. Also, since $\det q_k = \det d_k^{-1} \delta_k$, where $\delta_k = \det \la_k$, each entry of $q_k$ will have degree equal to $d(\delta_k)/n$. Hence, to prove the lemma we need to show that $d(\delta_k) = -n$.

But this follows from the fact that $N_k\Lambda_k = \rho_k^{-1}T(\rho_k)$, where $\Lambda_k$ is as in (\ref{Lambda}), must have determinant equals $1$ since $\rho_k$ does. Therefore,
\[
\det(N_k) \delta_k\delta_{k+1},\dots,\delta_{k+m-1} = 1.
\]
If we now apply the scaling, and having in mind that $d(N_k) = nm$, we get
\[
\mu^{nm}\det(N_k) \hat\delta_k\hat\delta_{k+1},\dots,\hat\delta_{k+m-1} = 1,
\]
where $\hat\delta_{k+1}$ is the scaled determinant. As show in \cite{MB2}, this system has a unique solution whenever $N$ and $m$ are coprime. But
\[
\hat\delta_{k+r} = \delta_k \mu^{-n},
\]
for all $k$, is clearly a solution. Hence $\delta_n$ are homogeneous and $d(\delta_k) = -n$. This concludes the proof.
\end{proof}

We are now in position to prove our main theorem.

\begin{theorem}\label{lambdath} The Grassmannian pentagram map on the moduli space $\Pm_N$ defined by (\ref{compa}) is invariant under the scaling (\ref{scaling1}).
\end{theorem}
\begin{proof}
We need to show that $T(a_k^i)$ are homogeneous, and $d(T(a_k^{2\ell})) = 0$, $d(T(a_k^{2\ell+1})) = 1$ for $\ell=0,\dots,s-1$. As in previous proofs, we will drop the subindex unless there could be some confusion.

Using (\ref{compa}), and denoting by $\ab$ the last column of $Q$, we can write $T(\ab)$ as
\[
T(\ab) =\Lambda^{-1}N^{-1} Q \Tc N \begin{pmatrix} O_n\\ \vdots\\ O_n\\ \la_{m+1}\end{pmatrix},
\]
or as the solution of the linear system of equations
\[
N\Lambda T(\ab) = Q \Tc N \begin{pmatrix} O_n\\ \vdots\\ O_n\\ \la_{m+1}\end{pmatrix}.
\]
Since we plan to use Cramer's rule once more, we will study the associated determinants.

To start with, we know that $\det(N\Lambda)$ is a homogeneous function of degree $nm-nm=0$. Define
\begin{eqnarray*}
D_{r,i}^j &=&\det(F\la, F_1\la_1,\dots,F_{r-1}\la_{r-1},F_{r,i}^{j,\la},F_{r+1}\la_{r+1},\dots,F_{2s-1}\la_{2s-1})\\&=& \det(F, F_1,\dots,F_{r-1},F_{r,i}^{j,\la},F_{r+1},\dots,F_{2s-1})\delta_1\dots\delta_{r-1}\delta_{r+1}\dots\delta_{2s-1},
\end{eqnarray*}
where $F_{r,i}^{j,\la}$ has all the columns equal to $F_{r}\la_{r}$, except for the $i$th column which is given by the $j$th column of $F_{2s}\la_{2s}$, that is, by $F_{2s}\la_{2s}^j = F_{2s}d_{2s}q_{2s}^j e_j$, with $d_{2s}$ and $q_{2s}$ as in (\ref{norm3}) and (\ref{ab}).

{\it Assume $r = 2\ell$}.

As in the first lemma, we can use (\ref{maineq}) to write down the determinants in terms of homogeneous components. For example, the $j$th column coming from $F_{2s}$ can be replaced by that of $G_{2s}$ in $F_{r,i}^{j,\la}$ and we can simplify the terms to the left of it, including the remaining columns of $F_{2\ell}$. We obtain
\begin{eqnarray*}
&&\det(F, F_1,\dots,F_{2\ell-1},F_{2\ell,i}^{j,\la},F_{2\ell+1},\dots,F_{2s-1})\\
&=&\det(F, F_1,\dots,F_{2\ell-1},G_{2\ell,i}^{j,\la},\hG_{2\ell+1}+f_{2\ell}^i(\al_{2\ell}^{2\ell+1})_i,G_{2\ell+2},\dots,\hG_{2s-1}+f_{2\ell}^i(\al_{2\ell}^{2s-1})_i).
\end{eqnarray*}
The block column $G_{2\ell,i}^{j,\la}$ is equal to $G_{2\ell}\la_{2\ell}$ except for the $ith$ column, which is equal to $G_{2s}\la_{2s}^j$. We can further use (\ref{maineq}) to substitute $f_{2\ell}^i$ with $g_{2\ell}^i$. Once we do that, we see that the  even orthogonal subspace is generated by the columns of $G_{2r}$, $r=0,\dots,s-1$ except for the extra column in $G_{2\ell}$, which in this case is generically covered by combination of columns in $G_{2s}\la_{2s}^j$. Thus, as before, we can remove the $g_{2\ell}^i$ terms from the determinant and simplify to the left of the $2\ell$ position.

We obtain
\[
D_{r,i}^j =\det(F, \hG_1, G_2\dots,\hG_{2\ell-1},G_{2\ell,i}^{j,\la},\hG_{2\ell+1},\dots,\hG_{2s-1}) \delta\delta_1\dots\delta_{2\ell-1}\delta_{2\ell+1}\dots\delta_{2s-1},
\]
with hats and non-hats alternating. All of the columns of $G_r$ are homogeneous of degree $1$ for any $r$. On the other hand, the columns of $G_{2\ell,i}^{j,\la}$ have degree zero since $G_{2\ell}$ has degree $1$ and $\la_{2\ell}$ has degree $-1$, so does $G_{2s}\la_{2s}^j = G_{2s}d_{2s}q_{2s}^j e_j$. Therefore,
\[
d(D_{r,i}^j) = n(m-1) - n(m-1) = 0.
\]
The $(j,i)$ entry of $T(a^{2\ell})$ is given by $\frac{D_{2\ell,i}^j}{D}$, and hence $d(T(a^{2\ell})) = 0$.

{\it Assume $r = 2\ell+1$.}

In this case, and always using (\ref{maineq}), the determinant
\[
D_{r,i}^j =\det(F\la, F_1\la_1,\dots,F_{2\ell}\la_{2\ell},F_{2\ell+1,i}^{j,\la},F_{2\ell+2}\la_{2\ell+2},\dots,F_{2s-1}\la_{2s-1}),
\]
can be further simplified replacing the $i$th column of $F_{2\ell+1,i}^{j,\la}$, given by $F_{2s} \la_{2s}^j$,  by
\[
h_{2s}^{j,\la} = G_{2s} \la_{2s}^j + F_{2\ell+1}\al_{2\ell+1}^{2s} \la_{2s}^j.
\]
We can then simplify the side of the determinant to the right of $F_{2\ell+1,i}^{j,\la}$ so that $F_{2r+1} \to \hG_{2r+1}$ and $F_{2r} \to G_{2r}+ f_{2\ell+1}^i\left(\al_{2\ell+1}^{2r}\right)_i$. After this simplification we can also substitute $f_{2\ell+1}^i$ with $\widehat{g}_{2\ell+1}^i$ and $F_{2\ell+1,i}^{j,\la}$  by $\hG_{2\ell+1,i}^{j,\la}$, where all columns of $\hG_{2\ell+1,i}^{j,\la}$ are equal to those of $\hG_{2\ell+1}\la_{2\ell+1}$, except for the $i$th column, given by
\[
g_{2s}^{j,\la} = G_{2s} \la_{2s}^j + \hG_{2\ell+1}\al_{2\ell+1}^{2s} \la_{2s}^j.
\]
We then continue to simplify the part of the determinant to the left of $\hG_{2\ell+1,i}^{j,\la}$. The result of all these simplifications is the determinant
\[
\det(F, \hG_1, G_2,\dots, G_{2\ell},\hG_{2\ell+1,i}^{j,\la},G_{2\ell+2}+\widehat{g}_{2\ell+1}^i\left(\al_{2\ell+1}^{2\ell+2}\right)_i, \hG_{2\ell+3},\dots, G_{2s-2}+\widehat{g}_{2\ell+1}^i\left(\al_{2\ell+1}^{2s-2}\right)_i, \hG_{2s-1}).
\]
First of all, notice that if we expand this determinant, the term in the expansion without any $g_{2\ell+1}^i$ is given by
\[
\det(F, \hG_1,\dots, G_{2\ell}, X^\la, G_{2\ell+2},\dots, \hG_{2s-1}),
\]
where $X^\la$ equals $\hG_{2\ell+1}\la_{2\ell+1}$ except for the $i$th column given by $\hG_{2\ell+1}\al_{2\ell+1}^{2s}\la_{2s}^j$. All columns of $\hG_{2\ell+1}\la_{2\ell+1}$ have degree zero, except for the $i$th column which has degree $1$ ($d(\hG_{2\ell+1}) = d(\al_{2\ell+1}^{2s}) = 1$ while $d(\la_{2s}^j) = -1$). Therefore, the degree of this determinant is $n(m-1) +1$.

If we now look at any of the terms in the expansion containing $g_{2\ell+1}^i$, we have
\[
\sum_{r=\ell+1}^{s-1}\sum_{p=1}^n\left(\al_{2\ell+1}^{2r}\right)_{i,p}\det(F, \hG_1, \dots, Y^\la, \dots, \hG_{2r-1}, G_{2r}^i+\widehat{g}_{2\ell+1}^i e_p^T, \hG_{2r-1}, \dots, G_{2s-2}, \hG_{2s-1}),
\]
 where $\left(\al_{2\ell+1}^{2r}\right)_{i,p}$ is the $(i,p)$ entry of $\al_{2\ell+1}^{2r}$, $G_{2r}^i$ has zero $i$th column and where $e_p$ is the standard canonical basis of $\R^n$.  The matrix $Y^\la$ is equal to $\hG_{2\ell+1}\la_{2\ell+1}$ except for the $i$th column, which is equal to $g_{2s}^{j,\la}$. We can further simplify $g_{2s}^{j,\la}$ to become $G_{2s} \la_{2s}^j$ since its odd term is generated by the other columns.

 Each column has degree $1$, except for those in $Y^\la$, which have degree $0$. Since each term is multiplied by  $\left(\al_{2\ell+1}^{2r}\right)_{i,p}$, of degree $1$, each term has degree $n(m-1)+1$, and so does the determinant. From here
 \[
 d(D_{2\ell+1,i}^j) = n(m-1)+1 - n(m-1) = 1.
 \]
 Since the $(j,i)$ entry of $T(a^{2\ell+1})$ is given by $\frac{D_{2\ell+1,i}^j}D$, we have
 \[
 d(T(a^{2\ell+1})) = 1,
 \]
 which concludes the proof of the theorem.
\end{proof}

This theorem allows us to define the Lax representation for the map $T$ on $\Pm_N$. Indeed, if we define
\[
Q_k(\mu) = \begin{pmatrix} O_n&O_n&\dots&O_n&a_k^0\\ I_n&O_n&\dots&O_n&\mu a_k^1\\O_n&I_n&\dots&O_n&a_k^2\\ \vdots&\ddots&\ddots&\ddots&\vdots\\ O_n&\dots&O_n&I_n&\mu a^{2s-1}_k\end{pmatrix},
\]
then there is a unique matrix $N_k(\mu)$ such that
\begin{equation}\label{Lax1}
T(Q_k(\mu)) = N_k(\mu)^{-1} \Lambda_k^{-1} Q_k(\mu) \Lambda_n N_k(\mu).
\end{equation}
The matrix $N_k(\mu)$, which is invariant and hence depends on $a_k^r$, is simply the matrix $N_k$ in (\ref{N}) rescaled by (\ref{scaling1}), that is
\[
N_k(\mu) = \mu(\rb_k, Q_k(\mu)\rb_{k+1}, Q_k(\mu)Q_{k+1}(\mu)\rb_{k+2},\dots, \left[Q_k(\mu)\dots Q_{k+2s-2}(\mu)\right]\rb_{k+2s-1}).
\]
(We can ignore the factor $\mu$ in front.) The system of equations
\[
T(\eta_k) = \eta_k N_k(\mu); \quad \quad \eta_{k+1} = \eta_k Q_k(\mu),
\]
has (\ref{Lax1}) as compatibility condition. Equation (\ref{Lax1}) must be independent from $\mu$ since $T$ is defined by its last column, and it is preserved by the scaling, while the rest of the entries are zero or $I_n$, and hence independent from $\mu$. Hence, this system is a standard {\it Lax representation for the Pentagram map on Grassmannian}, which can be used as usual to generate invariants of the map. Indeed, the conjugation class of the monodromy is preserved by the map $T:\Pm_N\to \Pm_N$, as we saw before. Since a representative of the class is given by $M = Q_0Q_1\dots Q_{N-1}$, we obtain the following theorem.
\begin{theorem}\label{Riemann}
The map $T:\Pm_N\to \Pm_N$ lies on the Riemann surface
\[
\det(Q_0(\mu)Q_1(\mu)\dots Q_{N-1}(\mu) - \eta I_{nm}) = 0
\]
with $\mu,\eta\in \CC$.
\end{theorem}

\section{Pentagram map on $\Gr(n,(2s+1)n)$}

\subsection{Definition of the map}
In this section we will define the Pentagram map for the Grassmannian $\Gr(n,(2s+1)n)$ for $s\geq 1$. Recall that the dimension of
$\Gr(n,(2s+1)n)$ is $2sn^{2}$.

Assume that $\wp = (l_k)$ is a twisted polygon on $\Gr(n,(2s+1)n)$ and let $X_k$ be any twisted lift. Let $\Pi_k$ be the unique $n(s+1)$ linear subspace containing the following subspaces
\[
\Pi_k = \langle X_k, X_{k+2}, \dots , X_{k+2s}\rangle.
\]

We define the pentagram map to be the map $T$ taking the polygon $\wp$ to the unique twisted polygon (with the same monodromy) whose $k$th vertex has a lift given by the intersections $\Pi_k \cap \Pi_{k+1}$. This map can be defined either on the space of polygons, or on the vertices. 
The map $T$ is well defined and independent from the lift $X$. In fact, from the Grassmann formula we get
\begin{align}
(2s+1)n=\dim (\Pi_{k-1}+\Pi_{k})&=\dim \Pi_{k-1}+\dim \Pi_{k}-\dim (\Pi_{k-1}\cap \Pi_{k}) \nonumber \\
&=(s+1)n+(s+1)n-\dim (\Pi_{k-1}\cap \Pi_{k}) \nonumber \\
&=2(s+1)n-\dim (\Pi_{k-1}\cap \Pi_{k}), \nonumber
\end{align}
which shows that $\dim (\Pi_{k-1}\cap \Pi_{k})=n$ for any $k$, and hence $\Pi_k\cap \Pi_{k+1}$ is a lift of a unique element in $\Gr(n, (2s+1)n)$, an element equals to the $k$th vertex of $T(\wp)$. As before, we will abuse notation and use $T$ equally for the map on polygons, on their lifts, on frames or on the moduli space.

%

Clearly the pentagram map is invariant under the action of the projective group (linear on lifts), and therefore one is able to write it as a map on the moduli space of Grassmannian polygons, as represented by the invariants we found in section 3. This is what we do next.

Let us consider a twisted normalized lift $V=(V_{k})$ of a regular $N$-gon, $\wp = (l_k)$ as in theorem \ref{invariantsth}. Then, using dimension counting, there exist $2s+1$ squared $n\times n$ matrices $a_{k}^{0}, a_{k}^{1},\ldots, a_{k}^{2s-1}, a_{k}^{2s}$ such that
\begin{equation}\label{4}
V_{k+2s+1}=V_{k}a_{k}^{0}+V_{k+1}a_{k}^{1}+\cdots+V_{k+2s-1}a_{k}^{2s-1}+V_{k+2s}a_{k}^{2s}.
\end{equation}
The blocks will be normalized so that $a_k^0$ is diagonal or Toeplitz, and the entries of $a_{k}^{2s}$ have a number of syzygies that relate them.

If the lift is twisted, then  $a_{k}^{i}$ will be $N$-periodic for $i=0,1,2,\cdots, 2s$; that is
\begin{equation}\label{41}
a_{k+N}^i = a_k^i,
\end{equation}
for any $k$.

Since $T(V_k) \in \Pi_{k+1}$, we can assume that there exists $c_k^i$, $n\times n$ matrices such that
\begin{equation}\label{8}
T(V_{k})=V_{k+1}c_{k}^{1}+V_{k+3}c_{k}^{3}+\cdots+V_{k+2s-1}c_{k}^{2s-1}+V_{k+2s+1}c_{k}^{2s+1},
\end{equation}
for all $k\in \mathbb{Z}$.

On other hand, we can replace $V_{k+2s+1}$ in the last term of (\ref{8}) by
\begin{equation}\label{9}
V_{k+2s+1}=V_{k}a_{k}^{0}+V_{k+1}a_{k}^{1}+\cdots+V_{k+2s-1}a_{k}^{2s-1}+V_{k+2s}a_{k}^{2s}.
\end{equation}
It follows that
\begin{align}\label{09}
T(V_{k})&=V_{k+1}c_{k}^{1}+V_{k+3}c_{k}^{3}+\cdots+V_{k+2s-1}c_{k}^{2s-1}+ \nonumber \\
&\,\,\,\,\,\,+(V_{k}a_{k}^{0}+V_{k+1}a_{k}^{1}+\cdots+V_{k+2s-1}a_{k}^{2s-1}+V_{k+2s}a_{k}^{2s})c_{k}^{2s+1}  \\
&=V_{k}a_{k}^{0}c_{k}^{2s+1}+V_{k+1}[c_{k}^{1}+a_{k}^{1}c_{k}^{2s+1}]+V_{k+2}a_{k}^{2}c_{k}^{2s+1} \nonumber \\
&\,\,\,\,\,\,+V_{k+3}[c_{k}^{3}+a_{k}^{3}c_{k}^{2s+1}]+\cdots\cdots+V_{k+2s-2}a_{k}^{2s-2}c_{k}^{2s+1} \nonumber \\
&\,\,\,\,\,\,+V_{k+2s-1}[c_{k}^{2s-1}+a_{k}^{2s-1}c_{k}^{2s+1}]+V_{k+2s}a_{k}^{2s}c_{k}^{2s+1}. \nonumber
\end{align}
Since we also have $T(V_{k})\in \Pi_{k}$, it follows that $c_k^{2\ell+1} = -a_k^{2\ell+1} c_k^{2s+1}$ for any $\ell=0,\dots,s-1$, and

\begin{equation*}
T(V_{k})= \left[ V_k a_k^0 + V_{k+2} a_k^2 + \dots + V_{k+2s-2} a_k^{2s-2} + V_{k+2s} a_k^{2s}\right]c_{k}^{2s+1}.
\end{equation*}

\begin{remark}{\rm As in the previous case, the matrix invariant matrix $c_k^{2s+1}$ has no apparent restrictions, but in fact, it is completely determined. In order to be able to define the pentagram map on the moduli space coordinates given by the matrices $a_k^j$, we need to guarantee that $\{T(V_k)\}_{k=1}^N$ is the lift of $T(l_k)$ as described by theorem
\ref{invariantsth}, as far as $T(\wp)$ is generic. As we showed in theorem \ref{invariantsth}, there is such a unique lift, and $c_k^{2s+1}= \la_k$ will be the proportional matrix that appears in the theorem. In fact, we have not shown that if $\wp$ is regular, so is $T(\wp)$. As it was the case with the original pentagram map (\cite{Sch1}), the map is only generically defined.) 
}
\end{remark}

\begin{example}For $s=1$, that is, on $\Gr(n,3n)$ the Pentagram map is
\begin{equation*}
T(V_{k})=(V_{k}a^{0}_{k} + V_{k+2} a_k^2)\la_{k} = (V_k, V_{k+1}, V_{k+2})\begin{pmatrix} a_k^0\\ O_n\\ a_k^2\end{pmatrix} \la_k = \rho_k \rb_k \la_k,
\end{equation*}
for any $k\in \mathbb{Z}$.
\end{example}

Now, define
\begin{equation}\label{111}
\rho_{k}=(V_{k}\;V_{k+1}\;\ldots \;V_{k+2s-1}\;V_{k+2s}),
\end{equation}
so that for any $k\in \mathbb{Z}$
\[
T(V_{k})=\rho_{k}\begin{pmatrix}
                                                                      a^{0}_{k} \\
                                                                      O_{n} \\
                                                                      a^{2}_{k}\\ O_n\\
                                                                      \vdots \\
                                                                       O_n\\
                                                                       a_k^{2s} \\
                                                                    \end{pmatrix}\la_{k}.
\]
As before, if $\rho_{k+1} = \rho_k Q_k$, and
\[
\rb_k = \begin{pmatrix} a_k^0\\ O_n\\ a_k^2\\ \vdots\\ a_k^{2s-2}\\ O_n\\ a_k^{2s}\end{pmatrix},
\]
 and if we extend $T$ to $\rho_k$ by applying it to each block-column, we can write
\begin{equation}\label{112}
T(\rho_{k})=\rho_{k}\begin{pmatrix} \rb_k \la_k& R_{k+1} \rb_{k+1}\la_{k+1}& R_{k+2} \rb_{k+2}\la_{k+2}&\cdots& R_{k+2s}\rb_{k+2s}\la_{k+2s}\end{pmatrix},
\end{equation}
where, if $Q_k$ is given as in (\ref{K}), then $R_{k+i}= Q_{k}Q_{k+1}\cdots Q_{k+i-1}$, for $i = 1,2,\dots$.

Now, as we did before, for any $k\in \mathbb{Z}$  define
\begin{equation*}
N_{k}=\begin{pmatrix} \rb_k & R_{k+1} \rb_{k+1}& R_{k+2} \rb_{k+2}&\cdots& R_{k+2s}\rb_{k+2s}
            \end{pmatrix}.
\end{equation*}
It follows that
\begin{equation}\label{Laxfree}
T(\rho_{k})=\rho_{k}N_{k} \Lambda_k, \hskip 3ex \rho_{k+1} = \rho_k Q_k,
\end{equation}
where
\[
\Lambda_k = \begin{pmatrix} \la_k & O_n&\dots& O_n\\ O_n&\la_{k+1}&\dots&O_n\\ \ddots&\ddots&\ddots&\ddots\\ O_n&\dots& O_n&\la_{k+2s}\end{pmatrix}.
\]
As before,  the compatibility condition of these two natural maps is given by
\[
\Tc(T(\rho_{k}))=\Tc(\rho_{k}N_{k})=\rho_{k+1}N_{k+1}=T(\Tc(\rho_k)) = T(\rho_{k+1})=T(\rho_{k}Q_{k})=T(\rho_{k})T(Q_{k}).
\]
Hence
\begin{equation*}
\rho_{k+1}N_{k+1}\Lambda_{k+1}=T(\rho_{k})T(Q_{k})=\rho_{k}N_{k}\Lambda_kT(Q_{k}),
\end{equation*}
for all $k\in \mathbb{Z}$. It shows that (\ref{Laxfree}), together with
\begin{equation}\label{20}
T(Q_{k})=N_{k}^{-1}\Lambda_k^{-1}Q_{k}N_{k+1}\Lambda_{k+1},
\end{equation}
holds true for any $k\in \mathbb{Z}$ and describes a discrete, parameter free, Lax representation for the map $T$ defined on the moduli space as represented by the invariants that appear in the last column of $Q_k$. Notice that from (\ref{112}) we now that $\la_k$ are also periodic, that is
$\la_{k+N}=\la_{k}$  for any $k$. And from the definition in (\ref{K}) so are both $Q_k$ and $R_k$.

\subsection{A Lax representation for the pentagram map on $\Gr(n,(2s+1)n$}

As we did for the even dimensional case, in this section we will prove that one can introduce a parameter $\mu$ in (\ref{Laxfree}) in such a way that  (\ref{20}) will be independent from $\mu$. This will define a true Lax representation that can be used for integration of the map. As it was done in \cite{MB2}, we will prove that the map $T$ is invariant under a scaling, this time given by
\begin{equation}\label{scaling2}
  \quad a_k^{2r+1} \to \mu^{-1+r/s} a_k^{2r+1}, r=0,\dots s-1 \quad a_k^{2r} \to \mu^{r/s} a_k^{2r}, r=0,1,\dots, s.
 \end{equation}
We will follow the same steps as in the even dimensional case.  The first steps involve proving that the map defined without the proportional matrix $\la_k$ is invariant under the scaling. We will then calculate the degree of $\la_k$ using these results and incorporate the proportional matrices $\la_k$ to the map to finally calculate the degree of $T(a_k^i)$.

First of all, notice that if, as before, we denote by $F_{k+r}$ the $r+1$ block-column of $N_k$, the analogous to Lemma \ref{Mainlemma} still holds true. We cite it here without proof, since the proof is identical.

\begin{lemma}
Let $F_k = \rb_k$ and $F_{k+\ell} = R_{k+\ell}\rb_{k+\ell}$, $\ell=1,\dots$ as above. Then, there exist $n\times n$ matrices $\al_i^j$ such that
\begin{equation}\label{maineq2}
F_{k+2\ell} = \sum_{r=1}^\ell F_{k+2r-1}\al_{2r-1}^{2\ell} + G_{k+2\ell}, \quad F_{k+2\ell+1} = \sum_{r=0}^\ell F_{k+2r}\al_{2r}^{2\ell+1}+\hG_{k+2\ell+1},
\end{equation}
for $\ell\ge 1$, where
\begin{equation}\label{Godd2}
\hG_{k+2\ell+1} = \pb_k\left(\Tc G_{k+2\ell}\right)_{2s}+\Gamma \Tc G_{k+2\ell}, \quad \al_{2r}^{2\ell+1} = \Tc\al_{2r-1}^{2\ell}, \quad \al_0^{2\ell+1} =\left(\Tc G_{k+2\ell}\right)_{2s},
\end{equation}
with
\begin{equation}\label{p2}
\pb_k = \begin{pmatrix}O_n\\a_k^1\\ O_n\\ a_k^3\\ O_n\\ \dots\\ a_k^{2s-1}\\ O_n\end{pmatrix},
\end{equation}
 and
\begin{equation}\label{Geven2}
G_{k+2\ell+2} = \Gamma \Tc\hG_{k+2\ell+1}, \quad \al_{2r+1}^{2\ell+2} = \Tc\al_{2r}^{2\ell+1}, \quad G_k = F_k = \rb_k.
\end{equation}
By $A_{2s}$ we mean the last $n\times n$ block entry of a matrix $A$.
\end{lemma}
The main difference with the even dimensional case is that the even block-columns $G_{2\ell}$ start now with a non-zero block (while before it started with a zero one), and generate a $(s+1)n$ dimensional subspace, orthogonal to those generated by the odd ones $\hG_{2\ell+1}$, which start with a zero block and generate a $sn$ dimensional subspace. Also, in this case the first two blocks of $G_{2\ell}$ are zero since $G_{k+2\ell} = \Gamma \Tc\hG_{k+2\ell-1}$ and $\Gamma$ shifts all the blocks once downwards.

Assume first that we drop the $\Lambda_k$ factor and define
\[
\Tb(V_k) = \rho_k \rb_k.
\]
Define also $c_k^i = \Tb(a_k^i)$, which is given by the compatibility formula below, the analogous to (\ref{companol})
\[
\Tb(Q_k) = N_k^{-1} Q_k N_{k+1}.
\]
  Then
\[
\Tb(\ab_k) = N_k^{-1} Q_k \Tc F_{k+2s} = N_k^{-1} F_{k+2s+1},
\]
which can be written as
\begin{equation}\label{lineareq2}
 N_k \Tb(\ab_k) = F_{k+2s+1}.
\end{equation}
Once more { $\Tb(\ab_k)$ can be interpreted as the solution of the linear equation (\ref{lineareq2})}.

\begin{theorem}\label{theorem1}
The matrices $c_k^i$ are homogenous with respect to the scaling (\ref{scaling2}), and
\[
d(c_k^{2\ell}) = \frac{\ell}s, \quad d(c_k^{2\ell+1}) = -1+\frac\ell s.
\]
\end{theorem}
Notice that the degree of $c_k^r$ coincides with that of $a_k^r$. Later we will show that $\la_k$ are all invariant under scaling and this theorem will essentially prove the invariance of the map under (\ref{scaling2}).

 Let us once more drop the subindex $k$ unless needed.
As we did in the previous case, we will work with determinants of the form $D_{r,i}^j$, simplifing them down to their homogeneous component, and calculating their degree. Because they are solution of (\ref{lineareq2}), each entry of $c^r$ will be a quotient of these determinants and $D$, and this way we will be able to determine their degree. The study comes in a number of lemmas.

\begin{lemma} The determinant
\[
D= \det (F, F_1, \dots, F_{2s}),
\]
is invariant under (\ref{scaling2}).
\end{lemma}
\begin{proof}
Using (\ref{maineq2}) we have that
\[
D = \det(F, \hG_1, G_2, \dots, \hG_{2s-1}, G_{2s}).
\]
This time all $G_{2r} = \Gamma \Tc \hG_{2r-1}$ have the first two blocks equal $O_n$; therefore, the blocks in the first row of $D$ are all zero, except for the first block which is the first block of $F = \rb$, i.e. $a^0$, and which is invariant under scaling. If we simplify using the first $n$ rows, we have the determinant of a matrix that
looks like
\[
\begin{pmatrix} A_{1,1} & O_n & A_{1,2}& O_n &\dots & A_{1,s}&O_n\\ O_n & B_{1,1} & O_n&B_{1,2}&\dots&O_n&B_{1,s}\\ A_{2,1} & O_n & A_{2,2}& O_n &\dots & A_{2,s}&O_n\\ O_n & B_{2,1} & O_n&B_{2,2}&\dots&O_n&B_{2,s}\\ \vdots&\dots&\dots&\dots&\dots&\dots&\dots\\ A_{s,1} & O_n & A_{s,2}& O_n &\dots & A_{s,s}&O_n\\O_n & B_{s,1} & O_n&B_{s,2}&\dots&O_n&B_{s,s}\end{pmatrix},
\]
where $A_{i,j}$ are the nonzero blocks of $\hG_{2j-1}$ and $B_{i,j}$ are the nonzero blocks of $G_{2j}$. Using $n^2 s (s-1)$ exchanges of rows and columns, this determinant can be easily transformed into
\[
\det\begin{pmatrix} A & O_{sn}\\ O_{sn}& B\end{pmatrix},
\]
with $A = (A_{i,j})$ and $B = (B_{i,j})$. Also, since $G_{2\ell} = \Gamma\Tc\hG_{2\ell-1}$, we have that $ B = \Tc A$ and $D = \det A\det \Tc A$.

We will next show that $d(A_{i,j}) = \frac {i-j}s$. This will imply, from the definition of determinant, that $d(\det A) = n \sum_{i=1}^s\sum_{j=1}^s \frac{i-j}s = 0$, concluding the proof.

Indeed, from (\ref{Godd2}) we have
\begin{equation*}
\hG_{2\ell+1} = \pb\left(\Tc G_{2\ell}\right)_{2s}+\Gamma \Tc G_{2\ell} = \pb\left(\Tc^2 \hG_{2\ell-1}\right)_{2s-1}+\Gamma^2 \Tc^2 \hG_{2\ell-1},
\end{equation*}
and using this we conclude that
\begin{equation}\label{A}
A_{k,\ell+1} = a^{2k-1}\Tc A_{s,\ell} + \Tc^2 A_{k-1,\ell}.
\end{equation}

The degrees of the nonzero blocks of $\pb$ are given by
\[
-1,\quad -1+\frac 1s,\quad -1+\frac 2s,\quad\dots, \quad -1 + \frac{s-1}s,
\]
 while the degree of $a^{2s}$ is $1$. Thus,
 \[
 d(A_{k,1}) = d(a^{2k-1} a^{2s}) = d(\Tc a^{2k-2}) = \frac {k-1}s.
 \]

We now use induction. Assume that $d(A_{i,j}) = \frac{i-j}s$ for all $i=1,\dots s$ and all $j<\ell$. From (\ref{A}) we have that, since $d(a^{2k-1}) = -1+\frac {k-1}s$,
\[
d(a^{2k-1} \Tc^2 A_{s,\ell}) = -1 +\frac{k-1}s+\frac{s-\ell}s = \frac{k-\ell-1}s; \quad d(\Tc^2A_{k-1,\ell}) = \frac{k-1-\ell}s,
\]
and so
\[
d(A_{k, \ell+1}) = \frac {k-(\ell+1)}s,
\]
concluding the proof of the lemma.
\end{proof}
Let us denote by $D_{r,i}^j$ the determinant given by
\[
D_{r,i}^j = \det( F, F_1, \dots, F_{r-1}, F_{r,i}^j, F_{r+1}, \dots, F_{2s}),
\]
where $F_{r,i}^j$ is the block-column obtained from $F_r$ by substituting the $i$th column, $f_r^i$, with the $j$th column of $F_{2s+1}$, $f_{2s+1}^j$.
\begin{lemma} Determinant $D_{2\ell+1,i}^j$ is homogeneous for all $i,j=1,\dots n$, and $\ell = 0,\dots, s-1$, and
\[
d(D_{2\ell+1,i}^j) = -1+\frac \ell s.
\]
\end{lemma}
\begin{proof}
As in the previous cases we will make heavy use of (\ref{maineq2}) to reduce the determinants to their homogeneous components before calculating their degree. First of all we will substitute the $i$th column of $ F_{2\ell+1,i}^j$ in
\[
\det( F, F_1, \dots, F_{2\ell}, F_{2\ell+1,i}^j, F_{2\ell+2}, \dots, F_{2s}),
\]
by that of
\(
\hG_{2s+1} \), $\hg_{2s+1}^j$, since the difference is a combination of even columns, which are all present in the determinant. We can then simplify the columns to the right of this and substitute them by either $\hG_{2r+1}$ or by
\[
G_{2r}+\hg_{2\ell+1}^i\left(\al_{2\ell+1}^{2r}\right)_i,
\]
where a super index indicate the column, a subindex the row. We can then substitute all the remaining columns of $F_{2\ell+1,i}^j$ and $F_r$, $r\le 2\ell$ by those of $G_r$ or $\hG_r$, depending on parity, to obtain
\[
D_{r,i}^j = \det( F, \hG_1, \dots, G_{2\ell}, \hG_{2\ell+1,i}^j, G_{2\ell+2}+\hg_{2\ell+1}^i\left(\al_{2\ell+1}^{2\ell+2}\right)_i,\hG_{2\ell+3}, \dots, G_{2s}+ \hg_{2\ell+1}^i\left(\al_{2\ell+1}^{2s}\right)_i),
\]
where $\hG_{2\ell+1,i}^j$ is equal to $\hG_{2\ell+1}$ except for its $i$th column which has been substituted by the $j$th column of $\hG_{2s+1}$. Since the missing column in $\hG_{2\ell+1}$ has been substituted by another odd column, if we expand this determinant all the terms that include any extra column $\hg_{2\ell+1}^i$ will vanish since we already have odd columns equal to half the dimension in their standard position. Therefore
\[
D_{r,i}^j = \det( F, \hG_1, \dots, G_{2\ell}, \hG_{2\ell+1,i}^j, G_{2\ell+2},\hG_{2\ell+3}, \dots, G_{2s}).
\]
Consider now the following product
\begin{equation}\label{product1}
\hg_{2s+1}^{j} = \hg_{2\ell+1}^i\otimes (\hg_{2\ell+1}^i)^{-1}\otimes \hg_{2s+1}^{j},
\end{equation}
where by $(\hg_{2\ell+1}^i)^{-1}$ we mean the vector whose entries are the inverses of the entries of $\hg_{2\ell+1}^{i}$, and where the product represents the individual entries product, as customary. We claim that $(\hg_{2\ell+1}^i)^{-1}\otimes \hg_{2s+1}^{j}$ is homogeneous, each entry with the same degree. Indeed, using the fact that $d(A_{i,j}) = \frac {i-j}s$ (recall that $A_{i,j}$, $i=1,\dots s$, are the nonzero blocks of $\hG_{2j-1}$, $j=1,\dots, s$),  we can see that
\begin{eqnarray*}
d((\hg_{2\ell+1}^i)^{-1}\otimes \hg_{2s+1}^{j}) &=& d(\hg_{2s+1}^{j}) - d((\hg_{2\ell+1}^i)\\ &=&
(\ast,-1 , \ast,-1+\frac1s , \ast, \dots, \ast,-1+\frac{s-1}s , \ast) \\&-& (\ast,-\frac\ell s , \ast,\frac{1-\ell}s , \ast, \dots, \ast,\frac{s-1-\ell}s , \ast)\\
&=&
(\ast,-1+\frac \ell s , \ast,-1+\frac \ell s , \ast, \dots, \ast,-1+\frac \ell s , \ast),
\end{eqnarray*}
where $\ast$ indicates the position of a zero block. Therefore, substituting $\hg_{2s+1}^{j}$ by (\ref{product1}) in the determinant, expanding the determinant using this column, and using the fact that $D$ is invariant under scaling gives us
\[
d(D_{2\ell+1,i}^j) = -1 + \frac \ell s,
\]
as claimed.
\end{proof}

\begin{lemma}

Determinant $D_{2\ell,i}^j$ is also homogeneous for all $i,j=1,\dots,n$, $\ell=0,\dots,s-1$ and
\[
d(D_{2\ell,i}^j) = \frac {\ell}s.
\]

\end{lemma}

\begin{proof}
The proof of this case is a bit more complicated. We need to look at the determinant
\[
D_{2\ell,i}^j = \det( F, F_1, \dots, F_{2\ell-1}, F_{2\ell,i}^j, F_{2\ell+1}, \dots, F_{2s}).
\]
As before, we simplify the $i$th column of $F_{2\ell,i}^j$, using (\ref{maineq2}). That is, $f_{2s+1}^j$ will be substituted by
\begin{equation}\label{col1}
\hg_{2s+1}^j + f_{2\ell}^i\left(\al_{2\ell}^{2s+1}\right)_{i,j}.
\end{equation}
We can then substitute the columns to the right of this one: we substitute $F_{2r+1}$ by $\hG_{2r+1}+f_{2\ell}^i\left(\al_{2\ell}^{2r+1}\right)_{i}$ and $F_{2r}$ by $G_{2r}$. We can also substitute $f_{2\ell}^i$ in the expression of $F_{2r+1}$ by $g_{2\ell}^i$. We then change the remaining columns in $F_{2\ell,i}^j$, including $f_{2\ell}^i$, by those of $G_{2\ell}$ (we call the resulting matrix $G_{2\ell,i}^j$), and substitute the $F$'s blocks to the left of  $G_{2\ell,i}^j$ by $G$-blocks. Then, column (\ref{col1}) becomes
\begin{equation}\label{col2}
\hg_{2s+1}^j + g_{2\ell}^i\left(\al_{2\ell}^{2s+1}\right)_{i,j}.
\end{equation}
 and the resulting determinant is given by
\[
\det( F, \hG_1, G_2 \dots, \hG_{2\ell-1}, G_{2\ell,i}^j, \hG_{2\ell+1}+g_{2\ell}^i\left(\al_{2\ell}^{2\ell+1}\right)_{i}, G_{2\ell+2},\dots, \hG_{2s-1}+g_{2\ell}^i\left(\al_{2\ell}^{2s+1}\right)_{i} G_{2s}).
\]
 If we now expand the determinant using the $i$th column of $G_{2\ell,i}^j$, i.e. (\ref{col2}), we have
\begin{eqnarray*}
D_{2\ell, i}^j &=& \left(\al_{2\ell}^{2s+1}\right)_{i,j} D \\
&+&\det( F, \hG_1, G_2 \dots, \hG_{2\ell-1}, {G}_{2\ell,i}^j, \hG_{2\ell+1}+g_{2\ell}^i\left(\al_{2\ell}^{2\ell+1}\right)_{i}, G_{2\ell+2},\dots, \hG_{2s-1}+g_{2\ell}^i\left(\al_{2\ell}^{2s+1}\right)_{i} G_{2s}),
\end{eqnarray*}
where ${G}_{2\ell,i}^j$ has as $i$th column $\hg_{2s+1}^j$ and $G_{2\ell}$ elsewhere.

Notice that, once more, if we expand this determinant, the terms with no $G_{2\ell}^i$ will be zero since we have more odd columns that are needed to generate the odd subspace (indicated with a hat) given that ${G}_{2\ell,i}^j$ contains one. Thus, the determinant expands as
\[
\sum_{r=\ell+1}^{s-1}\sum_{p=1}^n \left(\al_{2\ell}^{2r+1}\right)_{i,p}\det(F,\hG_1,G_2,\dots, G_{2\ell,i}^j,\hG_{2\ell+1},\dots,\hG_{2r+1}^{p}+g_{2\ell}^i e_p^T, G_{2\ell+1},\dots, G_{2s}),
\]
where $\hG_{2r+1}^{p}$ indicates $\hG_{2r+1}$ with a zero $p$ column. Each one of these terms can transformed by shifting the $p$ column of $\hG_{2r+1}^{p}+g_{2\ell}^i e_p^T$ and the $ith$ column of  $ G_{2\ell,i}^j$. The expansion becomes
\begin{equation}\label{expansion}
\sum_{r=\ell+1}^{s-1}\sum_{p=1}^n \pm \left(\al_{2\ell}^{2r+1}\right)_{i,p}\det(F,\hG_1,G_2,\dots, G_{2\ell},\hG_{2\ell+1},\dots,\hG_{2r+1}^{p}+\hg_{2s+1}^j e_p^T, G_{2\ell+1},\dots, G_{2s}),
\end{equation}
and now we are ready to calculate degrees. Substitute $\hg_{2s+1}^j$ by
\[
\hg_{2s+1}^j = \hg_{2s+1}^j\otimes (\hg_{2r+1}^p)^{-1} \otimes \hg_{2r+1}^p,
\]
as before, and notice that the degree of $\hg_{2s+1}^j\otimes (\hg_{2r+1}^p)^{-1}$ is the same for all of its block-entries and equal to $-1+\frac rs$
\[
d(\hg_{2s+1}^j\otimes (\hg_{2r+1}^p)^{-1}) = d(\hg_{2s+1}^j)- d(\hg_{2r+1}^p)= \begin{pmatrix}\ast\\\frac{1-(s+1)}s\\\ast\\\frac{2-(s+1)}s\\\vdots\\\frac{s-(s+1)}s\\\ast\end{pmatrix} - \begin{pmatrix}\ast\\\frac {1-(r+1)}s\\\ast\\\frac {2-(r+1)}s\\ \vdots\\\frac{s-(r+1)}s\\\ast\end{pmatrix}.
\]
We now need to calculate the degree of $\al_i^j$. This is simple from (\ref{Godd2}) and (\ref{Geven2}). We see that
\[
d(\al_0^{2r+1}) = \frac{s-r}s, \quad \quad d(\al_{2\ell}^{2r+1}) = d(\al_0^{2(r-\ell)+1}) = 1-\frac{r-\ell} s,
\]
\[ d(\al_{2\ell+1}^{2r}) = d(\al_0^{2(r-\ell-1)+1}) = 1-\frac{r-\ell-1}s.
\]
With this, every term in the expansion (\ref{expansion}) has degree
\[
-1+\frac rs + d(\left(\al_{2\ell}^{2r+1}\right)_{i,p}) = -1+\frac rs + 1-\frac{r-\ell}s = \frac \ell s,
\]
and since
\[
d\left(\left(\al_{2\ell}^{2s+1}\right)_{i,j} D\right) = d\left(\left(\al_{2\ell}^{2s+1}\right)_{i,j}\right) = 1-\frac{s-\ell}s = \frac \ell s,
\]
we conclude the proof of the lemma.
\end{proof}

After these results, the proof of theorem \ref{theorem1} is immediate since the entry $(j,i)$ of  $c^r$ is the quotient
\[
\frac{D_{r,i}^j} D,
\]
and so $d(c^r) = d(D_{r,i}^j)$, which coincide with the statement of the theorem.

Finally, the following theorem is a consequence of \ref{theorem1}.
\begin{theorem} If $N$ and $m$ are coprime, the map $T$ is invariant under the scaling (\ref{scaling2}).
\end{theorem}
\begin{proof}
To prove this result we simply need to prove that $\la_k$ are invariant under scaling. After this fact is proved, computations similar to those in the proof of theorem \ref{lambdath} will show that even as we introduce $\la_k$ in the different block columns of determinants $D$ and $D_{\ell,i}^j$, they do alter neither the homogeneity nor the degree of the determinants because they are invariant under the scaling and all columns of a block-column have the same degree (and so do their linear combinations). Therefore
\[
d(c_k^r) = d(T(a_k^r)) = d(a_k^r),
\]
for any $k$ and $r=0,\dots,m-1$. We will avoid further details of those computations since they are almost identical to those in \ref{lambdath} and the interested reader can reproduce them.

To show that $\la_k$ are invariant under scaling is not hard: if $N$ and $m$ are coprime, their determinants, $\delta_k = \det \la_k$, are the unique solution of
\[
\det N_k \delta_k\delta_{k+1}\dots\delta_{k+m-1} = 1,
\]
and since $\det N_k = D_k$, $\det N_k$ is invariant under the scaling (since $d(D_k) = 0$) and so are $\delta_k$ for any $k$. We next look at each factor in the splitting of $\la_k = d_k q_k$. As in the previous case, the factors $d_k$ are determined by the normalization of
$B_k = \Pi_m(c_k^0)$ as in (\ref{pi}), and since $c_k^0$ are invariant under the scaling, so will $d_k$.

Finally, $q_k$ are uniquely determined by equations of the form (\ref{uno}) and (\ref{dos}) for $c_k^{m-1}$, and by their determinants. Since $c_{i,j}^k = c^{m-1}_{k-m+1} c^{m-1}_{k-m+1}	\dots c^{m-1}_{N+k-m}$ is homogeneous, equations (\ref{uno})-(\ref{dos}) are scaling invariant. Furthermore, $\det q_k = \det d_k^{-1} \delta_k$ and both $\delta_k$ and $\det d_k$ are  scaling invariant, so is $\det q_k$. Therefore, $q_k$ are scaling invariant and so are $\la_k$.
\end{proof}

As in the even dimensional case, this theorem allows us to define the Lax representation for the map $T: \Pm_N\to \Pm_N$. Indeed, if we define
\[
Q_k(\mu) = \begin{pmatrix} O_n&O_n&O_n&\dots&O_n&O_n&a_k^0\\ I_n&O_n&O_n&\dots&O_n&O_n&\mu^{-1} a_k^1\\O_n&I_n&O_n&\dots&O_n&O_n&\mu^{\frac1s}a_k^2\\O_n&O_n&I_n&\dots&O_n&O_n&\mu^{-1+\frac 1s} a_k^3 \\\vdots&\ddots&\ddots&\ddots&\vdots&\vdots&\vdots\\ O_n&O_n&\dots&O_n&I_n&O_n&\mu^{-1+\frac{s-1}s} a^{2s-1}_k\\O_n&O_n&\dots&O_n&O_n&I_n&\mu a^{2s}_k\end{pmatrix},
\]
then we can define
\[
N_k(\mu) = (\rb_k(\mu), R_k(\mu) \rb_{k+1}, R_{k+1}(\mu)\rb_{k+2}(\mu), \dots, R_{k+m-2}(\mu)\rb_{k+m-1}(\mu)),
\]
where
\[
\rb_k(\mu) = \begin{pmatrix} a_r^0\\ O_n\\ \mu^{\frac 1s} a_r^2\\ O_n\\ \mu^{\frac 2s} a_r^4\\ \vdots\\ O_n\\ \mu a_k^s\end{pmatrix},
\]
and $R_{k+r}(\mu) = Q_{k}(\mu)Q_{k+1}(\mu)\dots Q_{k+r}(\mu)$.
The compatibility condition of the system
\[
\rho_{k+1} = \rho_kQ_k(\mu), \quad\quad T(\rho_k) = \rho_k N_k(\mu),
\]
will be given by
\begin{equation}\label{Lax2}
T(Q_k(\mu)) = N_k(\mu)^{-1} \Lambda_k^{-1} Q_k(\mu) \Lambda_n N_k(\mu).
\end{equation}
 Equation (\ref{Lax2}) must be independent from $\mu$ since $T$ is defined by its last column, and it is preserved by the scaling, while the rest of the entries are zero or $I_n$, and hence independent from $\mu$. Hence, this system is a standard {\it Lax representation for the Pentagram map on Grassmannian} in the case $m = 2s+1$, which can also be used as usual to generate invariants of the map. With this new scaling, theorem \ref{Riemann} is also true in the case when $m$ is odd.

\vskip 5ex
{\bf Acknowledgements}: This paper is supported by Mar\'\i~Beffa's NSF grant DMS \#1405722, by Felipe's CONACYT grant \#222870, and by the hospitality of the University of Wisconsin-Madison during Felipe's sabbatical year. R. Felipe was also supported by the {\it Sistema de ayudas para a\~nos sab\'aticos en el extranjero, CONACYT primera convocatoria 2014}.

\end{document}